\documentclass{article}

\usepackage{PRIMEarxiv}

\usepackage[utf8]{inputenc} 
\usepackage[T1]{fontenc}    
\usepackage{hyperref}       
\usepackage{url}            
\usepackage{booktabs}       
\usepackage{amsfonts}       
\usepackage{nicefrac}       
\usepackage{microtype}      
\usepackage{lipsum}
\usepackage{fancyhdr}       
\usepackage{graphicx}       
\graphicspath{{media/}}     
\usepackage{graphicx}
\usepackage{algorithm}
\usepackage{stackengine}
\usepackage{color}
\usepackage{caption}
\usepackage{cite}
\usepackage{xcolor}
\usepackage{amsmath,amsthm,amssymb,amsfonts,bm}
\usepackage{textcomp}
\usepackage{multirow}
\usepackage{subcaption}
\usepackage{float}
\usepackage{booktabs}

\usepackage[algo2e]{algorithm2e}
\usepackage[]{algpseudocode}

\makeatletter
\newcommand*{\rom}[1]{\expandafter\@slowromancap\romannumeral #1@}

\newtheorem{theorem}{Theorem}

\newtheorem{remark}{Remark}

\newtheorem{definition}{Definition}

\renewcommand{\v}[1]{\ensuremath{\mathbf{#1}}}

\usepackage{lineno}
\title{Perfect Gradient Inversion in Federated Learning: A New Paradigm from the Hidden Subset Sum Problem
}

\author{
  Qiongxiu Li$^{1}$, Lixia Luo$^{2}$, Agnese Gini$^{3}$, Changlong Ji$^{4}$\\ \textbf{Zhanhao Hu$^{5}$, Xiao Li$^{6}$, Chengfang Fang$^{7}$, Jie Shi$^{7}$, Xiaolin Hu$^{6}$} \\
  $^1$Aalborg University, $^2$Hunan University of Science and Technology, $^3$University of Luxembourg, \\$^4$Institut Polytechnique de Paris, $^5$UC Berkey, $^6$Tsinghua University, $^7$Huawei International, Singapore\\
}

\begin{document}
\maketitle
\begin{abstract}
Federated Learning (FL) has emerged as a popular paradigm for collaborative learning among multiple parties. It is considered privacy-friendly because local data remains on personal devices, and only intermediate parameters---such as gradients or model updates---are shared. Although gradient inversion is widely viewed as a common attack method in FL, analytical research on reconstructing input training samples from shared gradients remains limited and is typically confined to constrained settings like small batch sizes. In this paper, we aim to overcome these limitations by addressing the problem from a cryptographic perspective. We mathematically formulate the input reconstruction problem using the gradient information shared in FL as the Hidden Subset Sum Problem (HSSP), an extension of the well-known NP-complete Subset Sum Problem (SSP). Leveraging this formulation allows us to achieve perfect input reconstruction, thereby mitigating issues such as dependence on label diversity and underperformance with large batch sizes that hinder existing empirical gradient inversion attacks. Moreover, our analysis provides insights into why empirical input reconstruction attacks degrade with larger batch sizes. By modeling the problem as HSSP, we demonstrate that the batch size \( B \) significantly affects attack complexity, with time complexity reaching \( \mathcal{O}(B^9) \). We further show that applying secure data aggregation techniques---such as homomorphic encryption and secure multiparty computation---provides a strong defense by increasing the time complexity to \( \mathcal{O}(N^9 B^9) \), where \( N \) is the number of local clients in FL. To the best of our knowledge, this is the first work to rigorously analyze privacy issues in FL by modeling them as HSSP, providing a concrete analytical foundation for further exploration and development of defense strategies.
\end{abstract}



%

\section{Introduction}
The era of big data presents big challenges. Nowadays user data is routinely collected and stored on portable devices like phones and tablets \cite{anderson2015technology,poushter2016smartphone}.  
Processing such a sheer amount of data over numerous personal devices poses severe privacy concerns as the involved data often contains personal information such as location and personal profile, etc.   
In response to these challenges, FL has emerged as a promising solution \cite{konevcny2016federateda,konevcny2016federatedb,mcmahan2016federated}. FL facilitates collaborative model training across multiple devices while safeguarding personal data on local devices.  Typically, the FL framework involves numerous local participants, each equipped with a local training dataset, alongside a centralized server responsible for coordination. Participants transmit intermediate model updates, such as gradients derived from local training data, to the server, which then aggregates these updates to refine the global model.

The privacy-preserving nature of FL is currently under debate. The claim that  FL is privacy-preserving
relies on the assumption that sharing intermediate model updates will not leak users' private training data. This assumption is, however, being challenged by numerous studies. One of the earliest works on privacy issues of FL is the DLG (Deep Leakage from Gradient) attack \cite{zhu2019deep}, which shows that  FL systems are susceptible to gradient inversion attacks, i.e., reversing the original training data from the observed gradient information.  The core concept of the DLG attack involves optimizing synthesized data to align its gradients with those produced by the real training data. Subsequent studies have refined and enhanced the DLG attack by increasing the attack efficiency, improving the resolution of the reconstructed data, or addressing the limitation of small batch sizes \cite{zhao2020idlg,geiping2020inverting,yin2021see,boenisch2021curious,geng2023improved,fowl2021robbing,zhu2020r,wei2020framework,yang2022using,zhao2022deep,xu2022agic}. In addition to the above optimization-based approaches which empirically demonstrate the privacy risks, another line of work seek to analytically assess the privacy threat in FL.  Geiping et al.  \cite{geiping2020inverting}, Fowl et al. \cite{fowl2021robbing} and Boenisch et al. \cite{boenisch2021curious} show that the training data can be perfectly reconstructed under specific conditions. These conditions may include scenarios where only one sample is activated within a mini-batch or cases where attackers can maliciously alter model parameters or architectures.

Conversely, some researchers argue that FL is not inherently vulnerable to privacy attacks. According to Wainakh et al. \cite{wainakh2022federated}, the effectiveness of existing attacks largely depends on specific and sometimes impractical assumptions.   As an example, the performance of gradient inversion attacks is significantly impacted by the batch size, i.e., how many training samples are processed in one iteration of the training process. Most existing attacks are tested under limited conditions with small batch sizes, or they require unrealistic assumptions such as no two samples in a batch sharing the same class label.   Notably, when attacks are conducted with larger, more realistic batch sizes, such as 128 or 256, their effectiveness significantly deteriorates \cite{yin2021see,geng2023improved}. 

This ongoing debate about FL's vulnerability to privacy attacks in practical scenarios has sparked our interest and prompted further investigation.  Specifically, we aim to answer the following research question:  \textbf{Is it possible to mathematically analyze the difficulty of the input reconstruction attack stemming from gradient sharing in FL? For example,  how does this complexity change with respect to key parameters like batch size or the number of neurons?}

To answer this research question, we consider an MLP with ReLU \cite{agarap2018deep} activation function. We mathematically show that the gradients of the first layer produced by a mini-batch, can be represented as the product of a weight matrix and an input matrix comprising the training samples within the mini-batch.  This weight matrix can be further decomposed into two matrices: one representing the partial gradient of the loss with respect to the output of the first layer, and the other as a binary matrix resulting from the activation function.
By formulating the problem in this manner, we establish a connection between the problem of input reconstruction attack in FL  and a classical cryptographic problem known as the hidden subset sum problem (HSSP) \cite{coron2020polynomial,coron2021provably,gini2022hardness}. Leveraging this connection, we further utilize cryptographic tools to assess the difficulty of input reconstruction attacks in FL.  Notably, the complexity of HSSP is linked to the well-studied subset sum problem (SSP), which is known to be NP-complete \cite{dasgupta2008algorithms}. 

Our analysis indicates that as the batch size increases, reconstructing the training samples from the gradient information becomes more challenging, or at least more computationally demanding. Moreover, our analysis using HSSP sheds light on the critical influence of batch size on the difficulty of these attacks. For certain parameter settings,  the time complexity of HSSP solvers notably increases, scaling approximately as $\mathcal{O}(B^9)$, where $B$ denotes the batch size. This provides theoretical insight into why gradient inversion attacks are less likely to succeed in scenarios with larger batch sizes.  As for the number of neurons, it also impacts the time complexity; however, the attack efficiency can be optimized by subsampling an appropriate number of rows from the gradient matrix (see Section \ref{subsec.timeMB} for details).

As a marriage of both cryptography and FL, our work has significant impacts on both communities. On one hand, deep learning models often behave like a black-box lacking reasoning and explainability. Our work calls for more FL researchers to reinspect the privacy of FL models from a cryptographic view, such that a concrete and rigorous analysis can be conducted, e.g., the time complexity characterized by the batch size in our work. 
On the other hand, our work inspires how variations of cryptographic building blocks and hard problems such as HSSP can be constructed in practice and how their general cases can be verified in FL.

\subsection{Paper contribution}
The primary contributions of our paper are summarized as follows:
\begin{enumerate}
    \item \textbf{Approach novelty}: To the best of our knowledge, this is the first work that mathematically characterizes the gradient leakage in FL as a cryptographic problem, specifically through the lens of HSSP. This approach offers a new perspective on security in FL systems.
    \item \textbf{Cryptographic analysis}:  By utilizing tools from HSSP, we analytically demonstrate that the feasibility of reconstructing input samples from gradients is intricately linked to the batch size, with a derived time complexity of $\mathcal{O}(B^9)$. This finding substantiates the difficulty of input reconstruction in scenarios involving large batch sizes.
    \item \textbf{Resolving key bottlenecks}: Via HSSP we address two major limitations commonly encountered in traditional empirical gradient inversion attacks.  Firstly, it achieves perfect reconstruction of input data,  thereby mitigating evaluation sensitivity issues across various metrics \cite{almohammad2010stego,pambrun2015limitations}. Second, unlike prior methods, our technique does not rely on the assumption that training samples within a mini-batch must come from diverse labels or classes \cite{yin2021see,geiping2020inverting}. 
    \item \textbf{Strategic defense mechanisms}: With the analytical complexity, we further show that applying secure aggregation techniques, such as secret sharing \cite{Cramer2015} or homomorphic encryption \cite{paillier1999public} to share global gradients aggregated over all participants instead of local gradients of the individual participants, significantly enhances security. Specifically,  it would directly increase the time complexity from $\mathcal{O}(B^9)$ to $\mathcal{O}(N^9B^9)$, where $N$ denotes the total number of participants, thereby strengthening the defense against privacy breaches. 
\end{enumerate}

\subsection{Paper outline} 
The rest of the paper is organized as follows:  Section \ref{sec.pre} introduces the fundamentals of FL and the corresponding HSSP, which are essential for understanding the problem formulation discussed subsequently. Section \ref{sec.proForm} outlines our approach to formulating the problem of reconstructing input training data from observed gradients and establishes its connection to HSSP.   Section \ref{sec.lattice} introduces the fundamentals of lattices and describes the existing attacks for solving HSSP and the corresponding time complexity. Section \ref{sec.mhssp} explains how to extend HSSP attacks to multidimensional HSSP (mHSSP) attacks specifically for FL and briefly discusses typical defense strategies. Numerical results and limitations are discussed in Section \ref{sec.numRes}.  Section \ref{sec.prior} introduces the related work. Finally, conclusions are given in Section \ref{sec.conclu}.

\section{Preliminaries}\label{sec.pre}
In this section, we review the necessary fundamentals for the rest of the paper. 

\subsection{Notations}
In this work, lowercase letters $x$, lowercase boldface letters $\v x$, and uppercase boldface letters $\v X$ stand for scalars, vectors, and matrices, respectively.  
The $i$-th entry of the vector  $\v x$ is denoted either as $x_i$ or $\v x [i]$. Similarly, the $(i,j)$-th entry of the matrix $\v X$ is either denoted as $x_{i,j}$ or  $\v X[i;j]$. $\v X[i; :]$ and $\v X[: ;i]$ denote the $i$-th row and column of matrix $\v X$, respectively. $(\v x_1, \ldots,\v x_n)$ and $(\v x_1; \ldots;\v x_n)$ denote the row-wise and column-wise concatenation of the corresponding vectors.  $\v X^{\top}$ denotes the transpose of $\v X$. $\operatorname{diag}\{\v X\} $ denotes taking the diagonal elements of the square matrix $\v X$. $\odot$ denotes element-wise multiplication. $\v 1_r$  denotes the vector of all ones of size $r$ and similarly $\v I_r$ denotes the identity matrix of size $r$.  

\subsection{Fundamentals of machine learning}
\noindent\textbf{Neural networks:}
Let $f_{\v W}: \mathbb{R}^m \rightarrow\{1, \cdots, k\}$ be a $k$-class classifier, defined as a sequence of layers parameterized by trainable weights $\v W$. Each layer comprises a linear operation paired with a non-linear activation function such as the ReLU activation. A deep neural network often consists of many different layers. The main objective of the model $f_{\v W}$ is to map an input to its desired ground-truth label, denoted as $\v x_i$ and $l_i$. As a result, the model weights $\v W$ are updated through an iterative training process, e.g., the mini-batch Stochastic Gradient Descent (SGD). It repeats the following steps for modifying the initial $\v W$: 1) sample a mini-batch of size $B$ from the training data $\left\{(\v x_j, l_j)\right\}_{j=1}^B$; 2) execute a forward pass through the model to generate predictions for this mini-batch; 3) calculate the difference between the predictions and the ground-truth labels, known as the loss $\mathcal{L}$; 4) compute the gradient of $\mathcal{L}$ w.r.t. the weights, called the weight gradient $\v G$, and update the weights accordingly.

 \hspace*{\fill} 

\noindent\textbf{Federated learning:}
Suppose there are $N$ different users and each with their own local training dataset. They would like to cooperate together to jointly learn a shared model based on their datasets. Collecting the datasets from individual users centrally would be costly and raise privacy concerns due to the sensitive personal information contained in the local data. FL allows these users to cooperate without the need to share their raw data directly.   Let $f_{\v W}^{(t)}(\cdot)$ be the model with its weights $\v W^{(t)}$ at iteration $t$. FL works through the following steps: 
\begin{enumerate}
    \item Initialization: at iteration  $t=0$, the central server, denoted as $C$, randomly initializes the weights $W^{(0)}$ of the global model.
    \item Local model training: at each iteration $t$, the server first randomly selects a subset of users  and each user $i$ receives the model $f_{\v W}^{(t)}(\cdot)$ from the server $C$. Each selected user $u_i$ then calculates the local gradient $\v G_i^{(t)}$ for $f_{\v W}^{(t)}(\cdot)$ based on one mini-batch sampled from their local dataset $\left\{(\v x_j, l_j)\right\}_{j=1}^B$.
    \item Model aggregation: server $C$ collects the local gradients from the selected users and then aggregates them to obtain an updated global model. The aggregation is often done by weighted averaging and  typically uniform weights are applied, i.e., 
\begin{align}
\v G^{(t)}&=\frac{1}{N} \sum_{i=1}^{N} \v G_i^{(t)} \label{eq.gave}  \\
\v  W^{(t+1)}&=\v W^{(t)}-\eta \v G^{(t)} \nonumber.
\end{align}
The last two steps are repeated for many iterations until the global model converges or a certain stopping criterion is met. 
\end{enumerate}

\noindent\textbf{Gradient inversion attacks:}
The gradient inversion attack typically works by iteratively refining an estimate of the private input data to align with the observed gradients generated by such data. For example, in the DLG attack \cite{zhu2019deep}, we consider $\v G_i^{(t)}= \nabla f_{\v W}^{(t)}(\v x_i)$ as the local gradient calculated based on the data $\v x_i$ at iteration $t$, and $\v G_i^{(t),\prime}= \nabla f_{\v W}^{(t)}(\v x^{\prime}_i)$  as the synthetic  gradient generated by the estimated data $\v x^{\prime}_i$. The core strategy of DLG is to minimize the discrepancy between these two gradients, formulated as follows:
\begin{align}
\v x_i^{\prime *} =\underset{\v x_i^{\prime}}{\arg \min }\big\|\v G_i^{(t)}-\v G_i^{(t),\prime}\big\|^2,
\end{align}
where $\v x_i^{\prime *}$ denotes the reconstructed input data. 
Many variations of gradient inversion attacks have been proposed in the literature, enhancing the techniques and methodologies used in such attacks ~\cite{zhao2020idlg,geiping2020inverting,yin2021see,boenisch2021curious,geng2023improved,fowl2021robbing,zhu2020r,wei2020framework,yang2022using,zhao2022deep,xu2022agic}.

\subsection{Fundamentals of Hidden Subset Sum Problem}
Before introducing the HSSP, we first introduce a related problem called SSP. 

\begin{definition}\textbf{Subset Sum Problem (SSP)}
Let $B$ be an integer. Given $x_1, \ldots, x_B \in \mathbb{Z}$ and $t \in \mathbb{Z}$, compute the unknown weights $a_1, \ldots, a_B \in \{0, 1\}$ such that
\begin{align*}
   t=a_1 x_1+ a_2 x_2+\cdots+a_B x_B 
\end{align*}
if they exist.
\end{definition}
Note that SSP is a well-known NP-complete problem, thus currently it is unknown if a deterministic polynomial algorithm exists \cite{dasgupta2008algorithms}.  
HSSP is fundamentally related to SSP by assuming both the mixture weights and private data are unknown/hidden. The formal definition is as follows: 

\begin{definition}\textbf{Hidden Subset Sum Problem (HSSP)}
Let $Q$ be an integer, and let $x_1, \ldots, x_B$ be unknown integers in $\mathbb{Z}_Q$. Let $M$ be an integer and $\mathbf{a}_1, \ldots, \mathbf{a}_B \in \{0,1\}^M$ be unknown binary vectors. Let $\mathbf{h}=\left(h_1, \ldots, h_M\right) \in$ $\mathbb{Z}^M$ satisfy:
\begin{align}\label{eq.hssp}
    \mathbf{h}= \mathbf{a}_1 x_1+ \mathbf{a}_2 x_2+\cdots+\mathbf{a}_B x_B \quad \mod ~Q.
\end{align}
Given the modulus $Q$ and the sample vector $\mathbf{h}$, the goal is to recover the unknowns, i.e., both the vector $\mathbf{x}=\left(x_1, \ldots, x_B\right)$ and the vectors $\mathbf{a}_i$'s, up to a permutation of the $x_i$'s and $\mathbf{a}_i$'s.
\end{definition}
We refer to $\v x$ as the hidden private data, and $\mathbf{a}_i$s as hidden weight vectors, respectively. 

\begin{definition} \textbf{Hidden Linear Combination Problem (HLCP)}\label{def.hlcp}
Let $Q$ be a positive integer,  and let $x_1, \ldots, x_B$ be unknown integers in $\mathbb{Z}_Q$. Let $\mathbf{a}_1, \ldots, \mathbf{a}_B \in \{0,1, ... , c\}^M$ be unknown vectors for a given $c\in \mathbb{N}^{+}$. Let $\mathbf{h}=\left(h_1, \ldots, h_M\right) \in$ $\mathbb{Z}^M$ satisfy:
$$
\mathbf{h}= \mathbf{a}_1 x_1+ \mathbf{a}_2 x_2+\cdots+\mathbf{a}_B x_B \quad \mod ~Q.
$$
Given $Q$, $c$ and $\mathbf{h}$, the hidden linear combination problem  consists in recovering both the unknown vector $\mathbf{x}=\left(x_1, \ldots, x_B\right)$ and the unknown vectors $\mathbf{a}_i$'s, up to a permutation of the $x_i$'s and $\mathbf{a}_i$'s.
\end{definition}
Note that HLCP is a natural extension of HSSP where the coefficients of the hidden weight vectors $\v a_i$' lie in a discrete interval $\{0, ... , c\}$ instead of $\{0, 1\}$, for a given $c\in \mathbb{N}^{+}$.

\subsection{Threat model}
We assume that the central server is semi-honest (or passive), it is not expected to alter the model parameters maliciously, but it does have the capability to gather information throughout the learning process.  The primary aim of this central server is to reconstruct the input training samples using the information it collects. 

\section{Problem formulation}\label{sec.proForm}
Upon examining equation \eqref{eq.gave}, we observe that at each iteration $t$, server $C$ can collect local gradients $\{\v G_i^{(t)}\}_{i=1,\ldots,N}$. For each user $i$ among $\{1,\ldots,N\}$, the mini-batch $\left\{(\v x_j, l_j)\right\}_{j=1}^B$ generates the gradient $\v G_i^{(t)}$. The extent to which gradient $\v G_i^{(t)}$ reveals information about the individual training samples $\left\{(\v x_j, l_j)\right\}_{j=1}^B$ largely depends on the complexity of the model. 
In this context, we consider a MLP which consists of several fully connected layers. Mathematically, a fully connected layer is represented by a linear transformation followed by an activation function; here, we specifically consider the ReLU activation function. Let $u$ denote the dimension of the input $\v x_j$, and let $M$ denote the number of neurons in the first fully connected layer. We define the weight matrix as $\v W \in \mathbb{R}^{M \times u}$ and the bias vector as $\v b \in \mathbb{R}^{M}$. For the $m^{\text{th}}$ neuron, $\v w_m^T$ represents the corresponding row in the weight matrix $\v W$, and $b_m$ is the component in the bias vector $\v b$. For an input $\v x_j$, the output of the $m^{\text{th}}$ neuron, denoted as $y_{m,j}$, is given by:
\begin{align} \label{eq.ydef}
    y_{m,j}= \mathrm{ReLU}(\v w_m^T \v x_j+ b_m) \nonumber\\ 
    =\max (0, \v w_m^T \v x_j+ b_m).
\end{align}
Let $\v g_{\v w_m^T}$ and $g_{b_m}$ denote the gradient of  the loss $\mathcal{L}$ w.r.t. the weight $\v w_{m}^T$ and bias $b_m$, respectively. 
Given a mini-batch of samples  $\left\{(\v x_j, l_j)\right\}_{j=1}^B$, the gradient $\v g_{\v w_m^T}$  equals the average of all gradients computed for each of the data points that make up the mini-batch, i.e., 
\begin{align}\label{eq.gwm}
\v g_{\v w_m^T} &=\frac{1}{B} \sum_{j=1}^B \frac{\partial \mathcal{L}}{\partial y_{m,j}} \frac{\partial y_{m,j}}{\partial \mathbf{w}_m^T}. 
\end{align}
Similarly, $g_{b_m}$ is given by 
\begin{align}\label{eq.gbm}
g_{b_m} &=\frac{1}{B} \sum_{j=1}^B \frac{\partial \mathcal{L}}{\partial y_{m,j}} \frac{\partial y_{m,j}}{\partial b_m}. 
\end{align}

\subsection{A special case where $B=1$}
It has been demonstrated that a single input data point $\v x$ can be reconstructed from the gradients within a network consisting of biased fully-connected layers, which may be preceded by other (possibly unbiased) fully-connected layers \cite{geiping2020inverting}. The underlying principle is straightforward. Consider the scenario where $B=1$ (a single data sample per batch),  and find a neuron where $y_m=\mathbf{w}_m^T \v x+b_m$, the above \eqref{eq.gbm} becomes 
\begin{align*}
\frac{\partial \mathcal{L}}{\partial b_m}=\frac{\partial \mathcal{L}}{\partial y_m} \frac{\partial y_m}{\partial b_m}=\frac{\partial \mathcal{L}}{\partial y_m}\nonumber
\end{align*}
as $\frac{\partial y_m}{\partial b_m}=1$. Given the fact that $\frac{\partial y_m}{\partial \mathbf{w}_m^T}=\v x^T$, \eqref{eq.gwm} becomes
\begin{align}
\frac{\partial \mathcal{L}}{\partial \mathbf{w}_m^T}&=\frac{\partial \mathcal{L}}{\partial y_m} \frac{\partial y_m}{\partial \mathbf{w}_m^T}=\frac{\partial \mathcal{L}}{\partial y_m} x^T =\frac{\partial \mathcal{L}}{\partial b_m} \v x^T. \nonumber
\end{align}
Consequently, if there is at least one neuron $m$ for which $\frac{\partial \mathcal{L}}{\partial b_m} \neq 0$, the private data sample $\v x$ can be perfectly reconstructed using the gradients $\frac{\partial \mathcal{L}}{\partial b_m}$ and $\frac{\partial \mathcal{L}}{\partial \mathbf{w}_m^T}$.

As batch sizes in practical applications are rarely limited to a single sample, it is essential to consider a more general scenario where the batch size, $B$, is greater than 1.

\subsection{A compact formulation when $B>1$}
In scenarios where $B>1$, a more complex approach is required to handle the increased data volume. To formulate the problem compactly, we collect the gradients of all neurons across the entire batch. By stacking these gradients together, we have
\begin{align}\label{eq.BweightX}
\v G_w=
\left(\v g_{\v w_1^T},
\v g_{\v w_2^T},
\dots,
\v g_{\v w_M^T}\right)
\in \mathbb{R}^{M\times u}.
\end{align}
Additionally, we stack all data from the mini-batch:
\begin{align}
\v X= \left(
\v x_1^T,
\v x_2^T,
\dots,
\v x_B^T \right)\in \mathbb{R}^{B\times u}\nonumber.
\end{align}
Then we have the following main result:
\begin{theorem}\label{thm.main}
The shared gradient $\v G_w$ produced by the mini-batch data $\v X$ can be compactly represented as a scaled version of the product of a weight matrix $\v D$ and $\v X$, formulated as:
   \begin{align}
    \v G_w =\frac{1}{B}\v D\v X. \label{eq.Gw}
\end{align}
Here, $\v D$ can be further decomposed into: 
\begin{align}
    \v D=\v L \odot \v R  , \nonumber
\end{align}
where $\v L\in \mathbb{R}^{M\times B}$ represents the partial gradient of the loss with respect to the output of the first layer and $\v R\in \mathbb{R}^{M\times B}$ is a binary matrix resulting from the activation
function. 
Hence, the goal of input reconstruction is to recover the unknown training data $\v X$ from the observed gradient $\v G_w$.
\end{theorem}

\begin{proof}
We define \(\v L\) and \(\v Y_k\) for \(k \in \{1, 2, \ldots, u\}\) as follows:
\begin{align}
\v L=
\begin{pmatrix}
\frac{\partial \mathcal{L}}{\partial y_{1, 1}}&\frac{\partial \mathcal{L}}{\partial y_{1, 2}}&\ldots&\frac{\partial \mathcal{L}}{\partial y_{1, B}}\\
\frac{\partial \mathcal{L}}{\partial y_{2, 1}}&\frac{\partial \mathcal{L}}{\partial y_{2, 2}}&\ldots&\frac{\partial \mathcal{L}}{\partial y_{2, B}}\\
\vdots&\vdots&\ddots&\vdots\\
\frac{\partial \mathcal{L}}{\partial y_{M, 1}}&\frac{\partial \mathcal{L}}{\partial y_{M, 2}}&\ldots&\frac{\partial \mathcal{L}}{\partial y_{M, B}}
\end{pmatrix} \in \mathbb{R}^{M\times B}\nonumber.
\end{align}
\begin{align}
\v Y_k =
\begin{pmatrix}
\frac{\partial y_{1, 1}}{\partial w_{1,k}}&\frac{\partial y_{1, 2}}{\partial w_{1,k}}&\ldots&\frac{\partial y_{1, B}}{\partial w_{1,k}}\\
\frac{\partial y_{2, 1}}{\partial w_{2,k}}&\frac{\partial y_{2, 2}}{\partial w_{2,k}}&\ldots&\frac{\partial y_{2, B}}{\partial w_{2,k}}\\
\vdots&\vdots&\ddots&\vdots\\
\frac{\partial y_{M, 1}}{\partial w_{M,k}}&\frac{\partial y_{M, 2}}{\partial w_{M,k}}&\ldots&\frac{\partial y_{M, B}}{\partial w_{M,k}}
\end{pmatrix}
 \in \mathbb{R}^{M\times B} \nonumber.
\end{align}

With \eqref{eq.ydef} we have 
\begin{align}
\frac{\partial y_{m,j}}{\partial w_{m,k}}=\begin{cases}
    \v x_j[k],& \text{if } \mathbf{w}_m^T \v x_j+b_m >0\\
    \v 0,              & \text{otherwise} 
\end{cases}.\nonumber
\end{align}
$\v R\in \mathbb{Z}^{M\times B}$ is a binary  matrix  where its $(m,j)$-th entry is given by 
\begin{align} 
r_{m,j}=\begin{cases}
    1,& \text{if } \mathbf{w}_m^T \v x_j+b_m >0\\
    0,              & \text{otherwise}
\end{cases}. \nonumber
\end{align}
Take the $k$-th column of $\v X$, i.e., $\v  X[: ;k]$, and use it to further define matrix $\v X_k$ as
\begin{align}
\v X_k=
\begin{pmatrix}
\v  X[: ;k]^T \\
\v  X[: ;k]^T\\
\vdots\\
\v  X[: ;k]^T\\
\end{pmatrix}
 \in \mathbb{R}^{M\times B} \nonumber,
\end{align} 
thus $\v Y_k$ can be represented as 
\begin{align}
    \v Y_k= \v X_k \odot \v R \nonumber
\end{align}
As for \eqref{eq.BweightX}, we can represent the $k$-th column of $\v G_w$  as
\begin{align}
\v G_w[: ;k] 
&=\frac{1}{B} \operatorname{diag} \big( \v L \v Y_k^T \big) \nonumber\\
&=\frac{1}{B} \operatorname{diag} \big( \v L  (\v X_k^T \odot \v R^T)  \big) \nonumber\\
&\stackrel{\text{$(a)$}}{=} \frac{1}{B} \operatorname{diag} \big( (\v L \odot \v R) \v X_k^T \big)\nonumber\\
&\stackrel{\text{$(b)$}}{=} \frac{1}{B}(\v L \odot \v R)\v  X[: ;k],
\end{align}
where (a) uses the fact that the $i$-th diagonal element of $(\v L \odot \v R) \v X_k^T$ is $\langle \mathbf{L}[i; :]\odot \mathbf{R}[i; :], \v X_k[: ;i]\rangle=\langle \mathbf{L}[i; :], \v X_k[: ;i]\odot \mathbf{R}^{\top}[: ;i]\rangle,$  where the right side of the equation is the $i$-th diagonal element of $\mathbf{L}\cdot (\v X^{\top}_k\odot\mathbf{R}^{\top})$,
(b) uses the fact that all columns of  $\v X_k^T $ are identical.
Stacking all columns $\v G_w[: ;k]$s for $k\in\{1,2,\ldots,u\}$ together completes the proof. 
\end{proof} 

A similar rule follows when considering the bias, we thus have  
\begin{align}
\v g_b &=\frac{1}{B} (\v L \odot \v R)\v  1_B=\frac{1}{B} \v D\v  1_B \label{eq.Gb},
\end{align}
where 
\begin{align}\label{eq.BweightBias}
\v g_b=
\left(
g_{b_1},
g_{b_2},
\dots,
g_{b_M}\right) \in \mathbb{R}^{M}.
\end{align}

Hence, Theorem \ref{thm.main} implies that the shared gradient information $\v G_w$ can be formulated as the product of the unknown weight matrix $\v D$ and the unknown training input $\v X$, where recovering the latter is the goal of input reconstruction. In the subsequent section we will explore how this connects to mHSSP, providing a cryptographic lens to examine the challenges of reconstructing inputs from shared gradients in FL.

\subsection{Connection to HSSP}\label{subsec.conn}
We now explain how to cast the problem of inverting input training data from the shared gradient as an instance of HSSP. Before that we first introduce an extended version of HSSP for handling multidimensional data. This extension is crucial since the private data in machine learning is typically multidimensional.
\begin{definition}\textbf{Multidimensional Hidden Subset Sum Problem (mHSSP)}
Let $Q$ be an integer, and let $\{\v x_i\}_{i\in \{1,\ldots,B\}}$ be vectors in $\mathbb{Z}_Q^u$. Let $\mathbf{a}_1, \ldots, \mathbf{a}_B \in \mathbb{Z}^M$ be vectors with components in $\{0,1\}$. $\forall j=\{1,\dots,u\} $, let $\v h_j \in \mathbb{Z}^M $ satisfying:
\begin{align}
   \mathbf{h}_j= \mathbf{a}_1 x_{1,j}+ \mathbf{a}_2 x_{2,j}+\cdots+\mathbf{a}_B x_{B,j} \quad \mod ~Q \nonumber
\end{align}
Given the modulus $Q$ and the sample matrix $\mathbf{H}=[\v h_1; \ldots; \v h_u] \in$ $\mathbb{Z}^{M\times u}$, the goal is to recover the unknowns: i.e., both the vectors $\{\mathbf{x}_{i}\}_{i\in \{1,\ldots,B\}}$ and the weights vectors $\mathbf{a}_i$'s, up to a permutation of the $\v x_i$'s and $\mathbf{a}_i$'s.
\end{definition}
Further denote $\v X=[\v x_1^T,\ldots,\v x^T_B]\in \mathbb{Z}^{B\times u}$ as the matrix of  hidden private data and $\mathbf{A}=[\mathbf{a}_1;\dots;\mathbf{a}_B] \in \{0,1\}^{M\times B}$ as the hidden weight matrix.
Thus, the above can be compactly expressed as 
\begin{align} \label{eq.mhssp}
    \v H= \v A \v X \quad \mod ~Q
\end{align}
Accordingly, the definition of multidimensional HLCP (mHLCP) is the same as the above except each entry of the hidden weight matrix $\v A$ lies in a discrete interval $\{0, ... , c\}$ instead of $\{0, 1\}$. \\ \hspace*{\fill} 
We present the following observation as a remark to emphasize its theoretical implications:
\begin{remark} \textbf{Analogy between gradient-based input reconstruction and mHSSP}:
    Given the gradient information $\v G_w = \frac{1}{B}\v D\v X$, an input reconstruction attack aims to recover the unknown training input $\v X$ from the shared gradients $\v G_w$. We note that this input reconstruction problem closely parallels the definition of mHSSP or mHLCP, where the equation $\v H = \v A \v X \mod Q$ defines the problem, and the goal is to recover the hidden weight $\v A$ and input $\v X$ using the known information including $\v H$ and the modulus $Q$.
\end{remark}
We identify two primary distinctions between our setup and the classical mHSSP (and mHLCP): the latter is traditionally formulated within a modulo domain and operate on integers. To bridge these differences, common strategies include scaling floating-point numbers to integers and employing modular additive inverses to address negative values. For the modulo operation, we select a sufficiently large modulus $Q$ to ensure that the results of $\v H = \v A \v X \mod Q$ and $\v H = \v A \v X$ are equivalent.
To simplify the analysis, we assume that the weights in matrix $\v D$ are binary, that is, matrix $\v L$ is composed entirely of ones, and $\v D$ is equivalent to $\v R$. We will later demonstrate that, even under this simplified framework, mHSSP remains challenging to solve, particularly with a large batch size $B$.

\begin{remark} 
\textbf{Applicability of mHSSP to FL protocols sharing model weights instead of gradients}:
In FL, the protocols may vary in their approach to data sharing; while some opt to share gradients, others share model weights directly \cite{mcmahan2017communication,li2020federated,li2021model}. This adjustment modifies the gradient sharing equation \eqref{eq.gave} as follows \footnote{For simplicity, we assume that each update involves only one batch of training samples per iteration, as handling multiple updates with several batches would make it more difficult to reconstruct inputs.}: 
\begin{align*}
\v W_i^{(t+1)}&=\v W_i^{(t)}-\eta \v G_i^{(t)}  \\
\v  W^{(t+1)}&=\frac{1}{N} \sum_{i=1}^{N} \v W_i^{(t)} \nonumber.
\end{align*}
Given the initialized model $\v W^{(0)}$, the processes of sharing gradients and sharing model weights effectively converge to the same outcome. This is because knowing the gradients is sufficient to compute the model weights. 
\end{remark}
Hence, our analysis of reconstructing inputs from gradients is applicable across various sharing protocols, highlighting the broad applicability of our approach.

\section{HSSP and attacks}\label{sec.lattice}
To cryptographically analyze the input reconstruction from gradient sharing in FL using HSSP,  we first introduce the necessary fundamentals of lattices and then introduce existing lattice attacks for solving HSSP.  We start with the Nguyen-Stern attack (NS attack), the very first HSSP attack raised by Nguyen and Stern \cite{nguyen1999hardness}. After that we will introduce two advanced attacks based on the NS attack by Coron and Gini based on different principles \cite{coron2020polynomial,coron2021provably}. Finally, we will discuss the time complexities of these three attacks.   
\subsection{Fundamentals of lattices}
A lattice in $\mathbb{R}^M$ is defined as a discrete subgroup and can be defined as follows \cite{cassels1997introduction,micciancio2002complexity}:
\begin{definition} \textbf{Lattice}
Given $B$ linearly independent vectors $\mathbf{b}_1, \dots, \mathbf{b}_B$ in $\mathbb{R}^M(M\geq B)$, the lattice generated by these basis vectors is the set
\begin{equation*}
\mathcal{L}(\mathbf{b}_1,\dots,\mathbf{b}_B)=\left\{\sum_{i=1}^{B}x_i\mathbf{b}_i|x_i\in \mathbb{Z},i=1,\dots, B\right\},
\end{equation*}
where $\mathcal{L}$ represents the lattice, and $B$ is called its rank, generally denoted as $\dim (\mathcal{L})$ . If $B = M$, the lattice is called full-rank.
\end{definition}

Denote $\mathbf{B}$ as an $M\times B$ base matrix with $\mathbf{b}_1,\dots,\mathbf{b}_B$ as its columns, then
\begin{equation*}
\mathcal{L}(\mathbf{B})=\mathcal{L}(\mathbf{b}_1,\dots,\mathbf{b}_B)=\left\{ \mathbf{B}\cdot\v x|\v x\in \mathbb{Z}^B\right\}.
\end{equation*}
It can be demonstrated that two lattices $\mathcal{L}(\mathbf{B})$ and $\mathcal{L}(\mathbf{B'})$ are identical if and only if there exists an unimodular matrix $\mathbf{U} \in \text{GL}_B(\mathbb{Z})$ such that $\mathbf{B}\mathbf{U} = \mathbf{B'}$. The determinant of a lattice, $\mathcal{L}(\mathbf{B})$, is defined as $\sqrt{\det(\mathbf{B}^\top \mathbf{B})}$. It can be verified that if $\mathcal{L}(\mathbf{B}) = \mathcal{L}(\mathbf{B'})$, then $\det(\mathcal{L}(\mathbf{B})) = \det(\mathcal{L}(\mathbf{B'}))$. Geometrically, the determinant of a lattice is inversely proportional to its density; hence, a smaller determinant corresponds to a denser lattice \cite{micciancio2002complexity}.

In this paper, we focus primarily on integer lattices, i.e., sublattices of $\mathbb{Z}^M$. 
\begin{definition} \textbf{Orthogonal lattice}
Let $\mathcal{L}\subseteq \mathbb{Z}^M$ be a lattice. Its orthogonal lattice is given as
$$ \mathcal{L}^{\bot}:=\{\mathbf{y}\in \mathbb{Z}^M \arrowvert \forall \v x\in \mathcal{L},\langle \v x,\mathbf{y}\rangle=0\},$$
where $\langle\ ,\ \rangle$ is the inner product of $\mathbb{R}^M.$
Its orthogonal lattice modulo $Q$ is given as
$$ \mathcal{L}_Q^{\bot}:=\{\mathbf{y}\in \mathbb{Z}^M \arrowvert \forall \v x\in \mathcal{L},\langle \v x,\mathbf{y}\rangle\equiv 0 \mod Q\}.$$
\end{definition}

The completion of a lattice $\mathcal{L}$ is the lattice $\bar{\mathcal{L}}= \mathrm{Span}_{\mathbb{R}}(\mathcal{L})\cap \mathbb{Z}^M=(\mathcal{L}^{\bot})^{\bot},$ where $\mathrm{Span}_{\mathbb{R}}(\mathcal{L})=\{\mathbf{B}\cdot \v x | \v x\in \mathbb{R}^B\}$ for $\mathcal{L}=\mathcal{L}(\mathbf{B}).$ Note that $(\mathcal{L}^{\bot})^{\bot}$ contains the original lattice $\mathcal{L}$. In addition, $\dim (\mathcal{L}^{\bot})+\dim (\mathcal{L})=M$ \cite{nguyen1999hardness} and $\dim (\mathcal{L}_Q^{\bot})=M$ \cite{gini2022hardness}.

\begin{definition} \textbf{Shortest vector}
Let $\Vert \cdot \Vert$ denote the Euclidean norm. The shortest non-zero vector $\mathbf{v}$ in a lattice $\mathcal{L}$ determines the first minimum of the lattice, denoted as $\lambda_1(\mathcal{L}) = \Vert \mathbf{v} \Vert$.
\end{definition}

\begin{definition} \textbf{Successive minima}
This is a generalization of $\lambda_1$.
Let $\mathcal{L}$ be a lattice of rank $B$. For $i\in \{1,\dots,B\}$, the $i$-th successive minimum $\lambda_i(\mathcal{L})$ is defined by the minimum of maximum norm of any $i$ linearly independent lattice points.
\end{definition}

The so-called LLL and BKZ algorithms are important foundational tools in lattice cryptography, playing a crucial role in designing HSSP attacks.
 Below is a brief overview of each algorithm:
\begin{itemize}
\item \textbf{LLL}:  The LLL algorithm \cite{lenstra1982factoring}  is a polynomial-time lattice basis reduction technique. It takes a basis of a lattice $\mathcal{L}$ as input and outputs an LLL-reduced basis. An LLL-reduced basis is characterized by its relatively short and nearly orthogonal vectors. This reduction not only simplifies the structure of the lattice but also facilitates further cryptographic computations, such as finding shorter or closer vectors relative to the target lattice basis.
\item \textbf{BKZ}: 
The BKZ algorithm \cite{chen2011bkz}  is an enhanced lattice reduction method that extends the ideas of the LLL algorithm by segmenting the lattice basis into blocks. Adjusting the block-size parameter allows for a trade-off between computational efficiency and the accuracy of the reduction. A larger block size leads to a more reduced basis, which is crucial for achieving finer lattice structures.  More specifically, BKZ-2 can produce an LLL-reduced basis in polynomial time. However, with full block-size one can retrieve the shortest vector of the lattice in exponential time. 
\end{itemize}

\subsection{NS attack}\label{sec.NSattack}
The whole process of the NS attack includes two main steps:

\begin{itemize}
    \item \textbf{Step 1}: The first step involves an orthogonal lattice attack. Given that the known information in \eqref{eq.hssp} is $\v h$ (in Section \ref{subsec.mhssp} we will discuss how to extend to the multidimensional case for \eqref{eq.mhssp}), the orthogonal lattice attack begins by calculating the orthogonal lattice of $\v h$  modulus $Q$, i.e., $\mathcal{L}_{Q}^{\bot}(\mathbf h)$. Since $\v h$ is  a linear combination of the hidden weight vectors $\{\mathbf{a}_i\}_{i\in \{1,\dots,B\}}$, the orthogonal lattice $\mathcal{L}^{\bot}(\mathbf{A})$ is then contained within $\mathcal{L}_{Q}^{\bot}(\mathbf h)$.  
   Observing that $\mathbf{a}_i$'s are binary, it could be inferred that the first $M-B$ short vectors of the LLL basis of $\mathcal{L}_{Q}^{\bot}(\mathbf h)$ are likely to belong to $\mathcal{L}^{\bot}(\mathbf{A})$. After obtaining $\mathcal{L}^{\bot}(\mathbf{A})$, then compute its orthogonal  $(\mathcal{L}^{\bot}(\mathbf{A}))^{\bot}$ using the LLL algorithm again, as  $(\mathcal{L}^{\bot}(\mathbf{A}))^{\bot}$ contains the target lattice $\mathcal{L}(\mathbf{A})$. Note that the LLL algorithm is utilized twice: first to compute the LLL-reduced basis of $\mathcal{L}^{\bot}_Q(\v h)$, and second to compute the orthogonal lattice of $\mathcal{L}^{\bot}(\v A)$ \cite{nguyen2006merkle, chen2018computing}.
    \item \textbf{Step 2}: This step focuses on recovering the binary vectors $\mathbf{a}_i$' from a LLL-reduced basis of $(\mathcal{L}^{\bot}(\mathbf{A}))^{\bot}$, such that the hidden private data vector $\mathbf{x}$ can then be recovered. Notice that the short vectors found in Step 1 might not be sufficiently short,  the BKZ algorithm is then applied to discover shorter vectors.  Let $\{\v v_i\}_{i=1}^{B}$ denote the obtained short vectors in  $(\mathcal{L}^{\bot}(\mathbf{A}))^{\bot}$. 
    Given the fact that $\mathbf{a}_i$'s are binary, it could be proved with a high possibility that these $\v v_i$s are either the vectors $\{\mathbf{a}_i\}$'s themselves or their differences $\{\mathbf{a}_i-\mathbf{a}_j\}$'s. Assuming $\mathbf{a}_i$'s are the only binary vectors within $(\mathcal{L}^{\bot}(\mathbf{A}))^{\bot}$, 
     then select $B$ binary vectors from $\{\mathbf{v}_i\}\cup \{\mathbf{v}_i-\mathbf{v}_j\}\cup \{\mathbf{v}_i+\mathbf{v}_j\}$ as $\mathbf{a}_i$'s. After obtaining $\v A$, select a $B\times B$ sub-matrix $\mathbf{A}'$ from $\mathbf{A}$ with non-zero determinant modulus $Q$. Subsequently, $\mathbf{A}'\cdot \v x\equiv \mathbf{h}' \mod Q$ allows the recovery of the hidden private data as $\v x\equiv \mathbf{A}'^{-1}\mathbf{h}' \mod Q.$
\end{itemize}

The above NS attack is guaranteed to work with a good probability if the parameters $M,B,Q$ satisfy certain conditions. The underlying idea is that the short vectors of the LLL basis of $\mathcal{L}_{Q}^{\bot}(\mathbf h)$ should form a basis of $\mathcal{L}^{\bot}(\mathbf{A})$. Consider a vector $\mathbf{y}$ that is orthogonal modulo $Q$ to $\mathbf{h}$. We thus have
\begin{align}\label{eq.yh}
    \langle \mathbf{y},\mathbf{h}\rangle=x_1\langle \mathbf{y},\mathbf{a}_1\rangle +\cdots+x_B\langle \mathbf{y},\mathbf{a}_B\rangle \equiv 0 \mod Q,
\end{align}
which implies that the vector $$\mathbf{p}_{\mathbf{y}}=(\langle \mathbf{y},\mathbf{a}_1\rangle,\dots,\langle\mathbf{y},\mathbf{a}_B\rangle)$$ is orthogonal to the hidden private data vector $\v x=(x_1,\dots,x_B)$, i.e., $\mathbf{p}_{\mathbf{y}}\in \mathcal{L}_{Q}^{\bot}(\v x)$.
It follows that if the norm of   $\mathbf{p}_{\mathbf{y}}$ is less than the first minimum of $\mathcal{L}_{Q}^{\bot}(\v x)$, $\mathbf{p}_{\mathbf{y}}$ must be a zero vector, thereby confirming that $\mathbf{y} \in \mathcal{L}^{\bot}(\mathbf{A})$. Given that $\dim (\mathcal{L}^{\bot}(\mathbf{A})) = M-B$, an upper bound for the $M-B$ successive minima of $\mathcal{L}^{\bot}(\mathbf{A})$ should be less than a lower bound estimation for the first minimum of $\mathcal{L}_{Q}^{\bot}(\mathbf{x})$. 
To guarantee $\mathbf{y} \in \mathcal{L}^{\bot}(\mathbf{A})$, the following inequality is required: 
\begin{align}
\log Q>\iota MB+\frac{MB}{2(M-B)}\log M+\frac{B}{2}\log(M-B), \nonumber
\end{align}
where $0<\iota<1$ is decided by the so-called  LLL Hermite factor which controls the quality of the LLL-reduced basis. For further details, we refer the reader to \cite{gini2022hardness,coron2020polynomial}.


\subsection{CG attacks}
The BKZ algorithm, used in Step 2 of the NS attack algorithm to identify shorter vectors, contributes to the non-polynomial time complexity of the NS attack. To address this inefficiency, Coron and Gini \cite{coron2020polynomial,coron2021provably} have proposed two alternative approaches for Step 2. These methods can efficiently recover the binary vectors $\mathbf{a}_i$'s within a polynomial time, significantly improving the practicality of the algorithm. Below we will briefly introduce the main ideas of these two approaches.

 \hspace*{\fill} 
 
\noindent\textbf{The multivariate approach \cite{coron2020polynomial}:}
Instead of using the BKZ algorithm to recover the shorter vectors, the multivariate approach proposes to recover the hidden vectors $\mathbf{a}_i$'s using a multivariate quadratic polynomial system.  Let $\{\mathbf{u}_{j}\}_{j\in \{1,\dots,B\}}$ denote a basis of $(\mathcal{L}^{\bot}(\mathbf{A}))^{\bot}$, obtained as the output from Step 1 of the NS algorithm.  Stacking them together and denote as matrix $\mathbf{U}=[\mathbf{u}_1;\dots;\mathbf{u}_B]\in \mathbb{Z}^{M\times B}$, further define a matrix $\v W \in \mathbb{Z}^{B\times B}$ such that the following holds: 
    \begin{align}\label{eq.UWA}
        \v U \v W=\v A.
    \end{align}
    The goal is then to recover the unknown $\mathbf{W}$ and $\mathbf{A}$ based on the knowledge of $\mathbf{U}$.
     Given that $\v A$ is binary, the equation simplifies to: 
    \begin{align}
        \v U[i; :] \v W[: ;j]=A[i,j]\in \{0,1\}. \nonumber
    \end{align}
    As $0,1$ are roots of the quadratic equation $x^2-x=0$, , this relationship can be expressed as:
    \begin{align}
       \langle \v U[i; :],\v W[: ;j] \rangle^2-\langle  \v U[i; :],\v W[: ;j] \rangle=0,\nonumber
    \end{align}
    where $i=\{1,\ldots,M\}, j=\{1,\ldots,B\}$. 
Hence,  the above $MB$ quadratic equations form a multivariate quadratic system for unknown $\v W\in \mathbb{Z}^{B\times B}$. Solving the quadratic system allows for the recovery of $\v W$, and consequently, $\v A$.

It is important to note that the above multivariate quadratic polynomial system requires that $M\approx B^2$ instead of $M=2B$ in the NS algorithm. With a bigger $M$, the complexity of Step 1 would be larger (see time complexity analysis in Section \ref{subsec.time}). To avoid the high complexity, the authors proposed a blocks orthogonal lattice attack to reduce the complexity of Step 1 to  $\mathcal{O}(B^9)$.
\\ \hspace*{\fill} 

\noindent\textbf{The statistical approach \cite{coron2021provably}:} Recall \eqref{eq.UWA}, since the ranks of $\mathbf{U}$ and $\mathbf{A}$ are both $B$, thus $\mathbf{W}$ is invertible over $\mathbb{Q}$. Denote its inverse as matrix $\mathbf{V}=\mathbf{W}^{-1}$. We then have 
\begin{align}
    \mathbf{U}=\mathbf{A}\cdot \mathbf{V}\nonumber.
\end{align} 
We remark that the above has a similar form of the so-called continuous parallelepiped, which is defined as 
\begin{align}
 \mathcal{P}_{[-1,1]}(\mathbf{V})=\{\sum_{i=1}^{B}c_i\v V[i; :] | c_i\in [-1,1]\} \nonumber,
\end{align}
where $\mathbf{V}\in \text{GL}_B(\mathbb{R})$. Given the fact that $\v A$ is binary,  the product $\mathbf{A} \cdot \v V[i; :]$ falls within a discrete parallelepiped with binary coefficients:
\begin{align}
 \mathcal{P}_{\{0,1\}}(\mathbf{V})=\{\sum_{i=1}^{B}c_i\v V[i; :] | c_i\in \{0,1\}\} \nonumber.
\end{align}
Consequently, we have 
\begin{align}
    \{ \v U[i; :] \}_{i=\{1,\ldots,M\}}\subset \mathcal{P}_{\{0,1\}}(\mathbf{V})\subset \mathcal{P}_{[-1,1]}(\mathbf{V})\nonumber.
\end{align}
The Nguyen-Regev algorithm \cite{nguyen2006learning} is designed to return, with a significant probability, a good approximation of a row of $\mathbf{V}$ when given $poly(n)$ vectors uniformly sampled from $\mathcal{P}_{[-1,1]}(\mathbf{V})$. Coron and Gini first apply this Nguyen-Regev algorithm and subsequently make some modifications to ensure the recovery of the target binary matrix $\v A$.

\begin{table}[t!]
\centering
\begin{tabular}{@{}lccc@{}}
  \toprule
  \(M\) & \(\log Q\) & First LLL & Second LLL \\
  \midrule
  \(\gg B\) & \(\mathcal{O}(BM)\) & \(\mathcal{O}(M^7 \cdot B^2)\) & \(\mathcal{O}(M^9/B^2)\) \\
  \(B^2\) & \(\mathcal{O}(B^3)\) & \(\mathcal{O}(B^{16})\) & \(\mathcal{O}(B^{16})\) \\
  \(2B\) & \(\mathcal{O}(B^2)\) & \(\mathcal{O}(B^9)\) & \(\mathcal{O}(B^7)\) \\
  \(B+1\) & \(\mathcal{O}(B^2 \log B)\) & \(\mathcal{O}(B^9 \log^2 B)\) & \(\mathcal{O}(B^7 \log^2 B)\) \\
  \bottomrule
\end{tabular}
\vspace{2mm}
\caption{Modulus size and time complexity of the two LLL algorithms of Step 1.}
\label{tab.com}
\vskip -8pt
\end{table}

\subsection{Complexity analysis}\label{subsec.time}
Overall,  the complete time complexity is approximately $\mathcal{O}(B^9)$.

\section{$\v m$HSSP attacks for FL and defenses} \label{sec.mhssp}
Given that the input samples of machine learning are often of high dimension, we will explain how to adapt the previously introduced HSSP attacks to handle multidimensional inputs, thereby developing mHSSP attacks.
Subsequently, we explore a popular defense mechanism—secure data aggregation and analyze its impact on time complexity.
\subsection{From HSSP to mHSSP}\label{subsec.mhssp}
The input of mHSSP problem is a sample matrix  $\v{H}$ of dimension $M\times u$. Analogous to the scenario in \eqref{eq.yh} from Section \ref{sec.NSattack} where $u=1$, for appropriate choices of the parameters $u,M,B,Q$,  the short vectors of the LLL basis of $\mathcal{L}_{Q}^{\bot}(\mathbf H)$ should form a basis of $\mathcal{L}^{\bot}(\mathbf{A})$. Consider the vector $\mathbf{y}$ that is orthogonal modulo $Q$ to $\mathbf{H}$:
\begin{align*}
    \mathbf{y}\mathbf{H}\equiv \mathbf{yAX}\equiv \mathbf{p}_{\mathbf{y}}\mathbf{X} \equiv \mathbf{0} \mod Q,
\end{align*}
which implies that $\mathbf{p}_{\mathbf{y}}\in \mathcal{L}_{Q}^{\bot}(\v X)$. Therefore, if the norm of  $\mathbf{p}_{\mathbf{y}}$ is less than the first minimum of $\mathcal{L}_{Q}^{\bot}(\v X)$, $\mathbf{p}_{\mathbf{y}}$ must be a zero vector. 
This insight enables us to extend the reasoning used for $u=1$. Specifically, we can still compute a basis of $(\mathcal{L}^{\bot}(\v A))^\bot$ as in Step 1 of the NS attack, except that here, we replace $\mathcal{L}_Q^{\bot}(\mathbf{h})$ with $\mathcal{L}_Q^{\bot}(\mathbf{H})$.

\subsection{Constructing $\mathcal{L}_Q^{\bot}(\mathbf{H})$ in mHSSP attacks}
We now explain how to modify the above attacks for generating $\mathcal{L}_Q^{\bot}(\mathbf{H})$. 
 One effective modification involves transitioning from using $\mathcal{L}_Q^{\bot}(\mathbf{h})$ to $\mathcal{L}_Q^{\bot}(\mathbf{H})$. Traditionally, the orthogonal lattice $\mathcal{L}_Q^{\bot}(\mathbf{h})$ is constructed in the following manner:

Recall $h_1$ denotes the first entry of vector $\v h$, for simplicity we assume the greatest common divisor $\gcd(h_1, Q)=1$ (see Appendix \ref{app.lattice} for a more generalized discussion).
In addition, $\langle (Q,0,\dots,0), \mathbf{h}\rangle\equiv 0 \mod Q.$ Define $\v h'=(h_2,...,h_M)$ and a basis of $\mathcal{L}_Q^{\bot}(\mathbf{h})$ can be constructed as 
$$\begin{pmatrix} Q & \mathbf{0}\\ -\mathbf{h}'\cdot (h_1^{-1} \mod Q) & \mathbf{I}_{M-1} \end{pmatrix}^T,$$ 
as it can be verified that each column vector is orthogonal to $\v h$ modulo $Q$.  

As in practice, the dimension of the hidden private data $u$ can be very large, for example $u=3072$ using CIFAR-10 \cite{cifar} as the training dataset. In the attacks, we can make use of these  $u$ columns for generating the orthogonal lattice.  For example, we can take a few columns, say $r\leq u$ , to construct $\mathcal{L}_Q^{\bot}(\mathbf{H})$, i.e. the column dimension of $\mathbf{H}$ is $r$. For simplicity,  assume $\mathbf{H}[1,\dots,r; :]\in \mathbb{Z}^{r\times r}$ is invertible modulus $Q$ and denote the inverse as $\v H_{r}^{-1}$.
Hence, we then construct a basis of $\mathcal{L}_Q^{\bot}(\mathbf{H})$ with the following:
$$\begin{pmatrix} Q \mathbf{I}_r & \mathbf{0}\\ -\mathbf{H}[r+1,\dots,M; :]\cdot \v H_{r}^{-1} & \mathbf{I}_{M-r}\end{pmatrix}^T.$$



\subsection{Applying secure data aggregation before sharing gradients} \label{subsec.sa}
There are various ways to enhance the privacy of FL, one of the most widely adopted methods is to apply secure data aggregation techniques to securely compute the average of the local gradients without revealing them \cite{xu2019verifynet}. These techniques allow for the secure computation of the average of local gradients without disclosing the individual gradients \cite{xu2019verifynet}. Notable techniques include secure multiparty computation \cite{Cramer2015}, which enables multiple parties to jointly compute a function's output from their private inputs without revealing the inputs themselves, and homomorphic encryption \cite{paillier1999public,damgaard2012multiparty,gentry2009fully}, which ensures privacy by enabling computations to be performed directly on encrypted data.

After applying secure data aggregation, the available information to the server is thus the global average of the local gradients of all clients. Assuming there are $N$ clients, let $\v G_{w,i}$, $\v D_i$ and $\v X_i$  denotes the local gradient, the hidden weight matrix, and input training data of client $i$ where $i = \{1,\ldots, N\}$, respectively. The global average of the local gradients of \eqref{eq.Gw}, denoted as  $\v G_{w,\mathrm{ave}}$, can then be represented as follows:
\begin{align}
\v G_{w,\mathrm{ave}}&=\frac{1}{N}\sum_{i=1}^{N} \v G_{w,i} \nonumber\\
&=\frac{1}{NB}\begin{pmatrix}
\v D_1;&\v D_2;&\ldots;&\v D_N
 \end{pmatrix}
\begin{pmatrix}
\v X_{1}\\
\v X_{2}\\
\vdots\\
\v X_{N}\\
\end{pmatrix} \nonumber.
\end{align}
That is, applying secure data aggregation implies that the hidden inputs row dimension increases by $N$ times. Hence, the attack complexity directly increases from $\mathcal{O}(B^9)$ to  $\mathcal{O}(N^9B^9)$.  In Fig. \ref{fig.timeHssp} of the following section, we validate this result by showing how the time complexity of HSSP attacks increases with the underlying batch size.

\begin{figure*}[t!]
  \begin{subfigure}[t]{1.0\linewidth}
    \centering\includegraphics[width=130mm]{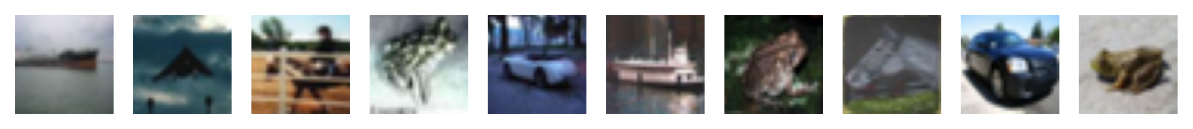}
    \centering\includegraphics[width=130mm]{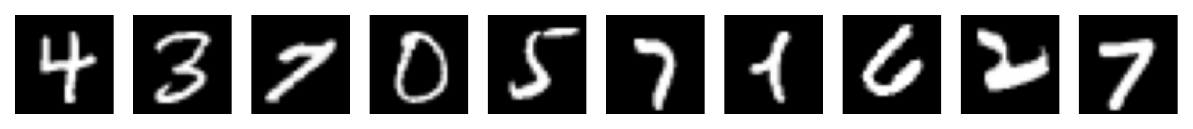}
    \centering\includegraphics[width=130mm]{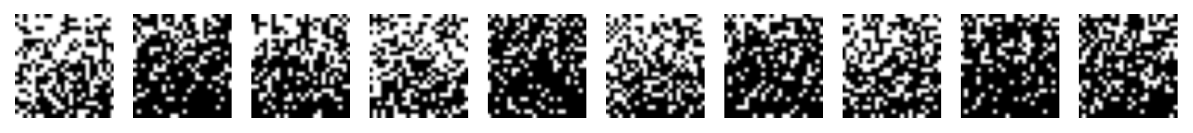}
    \captionsetup{justification=centering}
    \caption{Ground truth: $10$ training samples randomly selected from three datasets \\ including CIFAR-10 (top panel), MNIST (middle panel) and Purchase (Bottom panel\footnote{Note that the Purchase is a tabular dataset with each sample having 600 binary features, we resize it as $25\times 24$ for illustration purpose.}).}
  \end{subfigure}
    \begin{subfigure}[t]{1.0\linewidth}
    \centering\includegraphics[width=130mm]{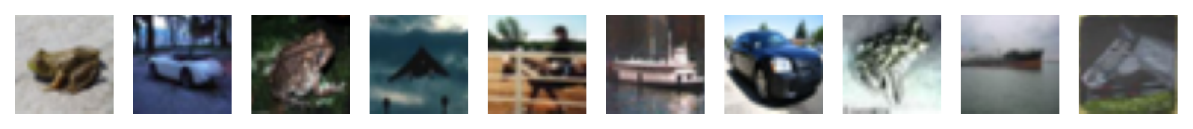}
    \centering\includegraphics[width=130mm]{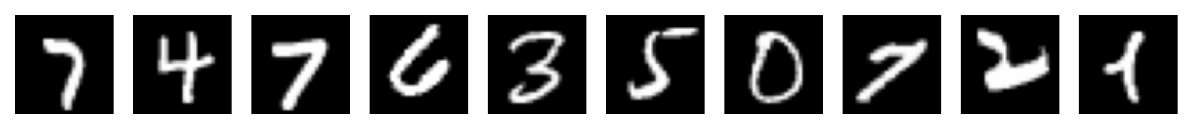}
    \centering\includegraphics[width=130mm]{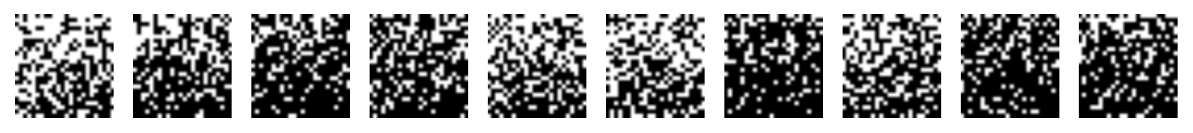}
    \caption{Reconstructed samples using the NS attack for each dataset.}
  \end{subfigure}
    \begin{subfigure}[t]{1.0\linewidth}
    \centering\includegraphics[width=130mm]{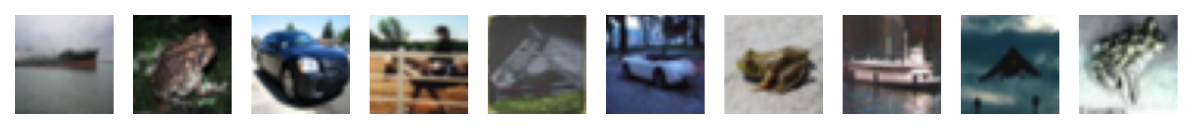}
    \centering\includegraphics[width=130mm]{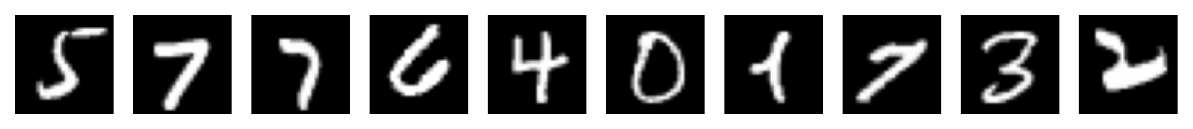}
    \centering\includegraphics[width=130mm]{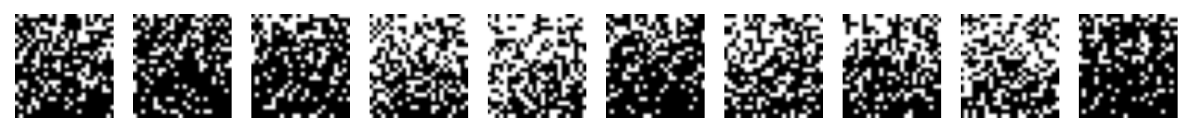}
    \caption{Reconstructed samples using the multivariate attack for each dataset.}
  \end{subfigure}
    \begin{subfigure}[t]{1.0\linewidth}
    \centering\includegraphics[width=130mm]{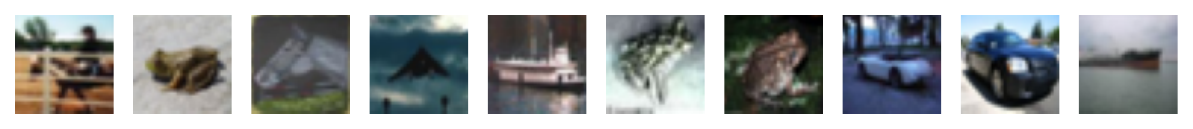}
    \centering\includegraphics[width=130mm]{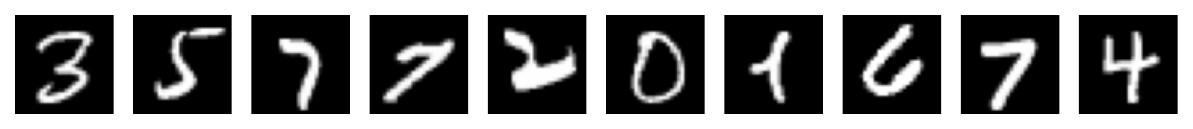}
    \centering\includegraphics[width=130mm]{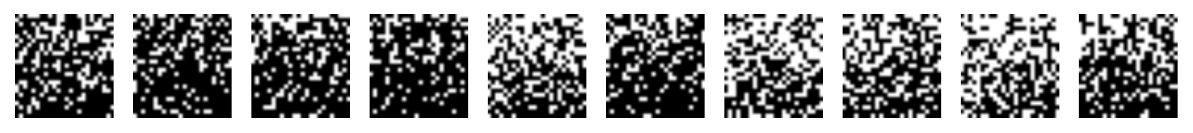}
    \caption{Reconstructed samples using the statistical attack for each dataset.}
  \end{subfigure}
\caption{Visualization examples of input reconstruction results using three mHSSP attacks for three datasets, respectively.}
\label{fig.fl1}
\end{figure*}

\section{Numerical results and analysis}\label{sec.numRes}
We now proceed to present the numerical results of mHSSP attacks in FL and discuss their implications. 
\subsection{Experimental setup}
We now evaluate the input reconstruction performance of mHSSP attacks using various datasets: two image datasets, MNIST~\cite{deng2012mnist}, CIFAR-10~\cite{cifar}, and one non-image dataset Purchase~\cite{shokri2017membership} (extended experiments using CIFAR-100 and Imagenet can be found in Fig. \ref{fig.imgnet} of Appendix \ref{app.con}). The MNIST dataset comprises 60,000 hand-written digital images of 28 $\times$ 28 pixels each, while CIFAR-10 includes 50,000 natural images of 32 $\times$ 32 pixels, both distributed across 10 classes. The Purchase dataset, derived from the ``Acquire Valued Shoppers'' challenge on Kaggle, contains shopping histories transformed into 197,324 samples with 600 binary features each~\cite{shokri2017membership}.

In our experiments, we simulate a federated setting with 10 clients. Each training dataset is randomly split into 10 folds and each client holds one fold. The model architecture is primarily based on MLPs to align with our theoretical analysis. We demonstrate the generality of our results across different MLP configurations by employing two specific MLP designs: for training CIFAR-10 and MNIST we use an MLP with 3 layers and (500, 100, 10) neurons in each layer respectively; for the Purchase dataset, we employ MLPs with four layers (1024, 512, 256, and 100 neurons in each layer).  
The learning rate for all experiments is set at 0.03 unless specified otherwise. The test accuracies of each model for CIFAR-10, MNIST, and Purchase are around $54\%, 98\%$, and $90\%$, respectively. It is important to note that while MLPs generally perform not well for complex natural image datasets like CIFAR-10, they are used here only for illustrative purposes.

\subsection{Visualization of mHSSP attacks' results}
To demonstrate the attack results, we randomly subsample the gradient information $\v G_w$ to the dimension of $300\times 20$ as the input (scaled up by $B$) of mHHSP attacks. In Fig. \ref{fig.fl1} we show the reconstructed samples using the three described mHSSP attacks under the case of $B=10$ for three datasets: CIFAR-10, MNIST, and Purchase, respectively (see Fig. \ref{fig.imgnet} in Appendix \ref{app.con} for results of using CIFAR-100 and Imagenet datasets). The top panel (a) shows 10 ground truth samples randomly sampled from each dataset. The subsequent plots demonstrate the reconstructed samples using each attack: (b) the NS attack, (c) the multivariate attack, and (d) the statistical attack.   
As we can see, all mHSSP attacks can successfully recover the original input private samples, except the fact that the order is permuted. Note that in this case, the 
reconstruction is perfect, i.e., the Euclidean distance between the reconstructed samples and the original samples is zero.

\subsection{Perfect input reconstruction} To evaluate the quality of reconstructed inputs, several widely used metrics include Mean Squared Error (MSE), Peak Signal-to-Noise Ratio (PSNR), and the Structural Similarity Index Measure (SSIM), with the latter two are mainly used for image datasets.  For instance, SSIM, which ranges from -1 to 1, assesses the similarity between reconstructed images and their original counterparts. A score of $1$ denotes perfect resemblance, while a score of $0$ indicates no correlation.   
However, these metrics are often considered limited and unreliable in accurately reflecting sample similarity, as highlighted in several studies \cite{almohammad2010stego,pambrun2015limitations}. Our approach alleviates the issue of choosing appropriate similarity metrics by achieving perfect input reconstruction, where the reconstructed inputs are identical to the original ones, thereby providing a definitive measure of reconstruction accuracy. 

To further validate this, 
in Tab. \ref{tab.mse} we show the overall input reconstruction performances of three mHSSP attacks across CIFAR-10, MNIST, and the Purchase datasets.  In each experiment, we set the batch size as $10$ and randomly select $500$ samples for reconstruction. As demonstrated,  both SSIM$=1$  and MSE$=0$ (PSNR is not displayed as it becomes infinite if  MSE is zero), confirm that the reconstructed inputs are exactly the same as the original training samples.

\subsection{No restrictive assumption for samples with repeated labels} As mentioned before, traditional gradient inversion attacks typically require that samples within a mini-batch should not have repeated labels, as this can lead to similar reconstructed images for the same label \cite{yin2021see,geiping2020inverting}, otherwise the reconstructed images for the same label would be quite similar. In contrast, our approach does not necessitate this restrictive assumption. To demonstrate this, we conducted experiments with 40 randomly selected samples all labeled as ``cat'' from the CIFAR-10 dataset, and tested the input reconstruction results using the three mHSSP attacks. The results are demonstrated in Fig. \ref{fig.fl2}, we can see that though all images are from the same class, all the input samples can be reconstructed perfectly except that their order is permuted. Hence, this confirms the benefits of using HSSP to analyze the privacy leakage in FL.

\begin{table}[t!]
\centering
\begin{tabular}{@{}lccc@{}}
  \toprule
   & \multicolumn{3}{c}{Dataset} \\
   \cmidrule(l){2-4} 
   & CIFAR-10 & MNIST & Purchase \\
  \midrule
NS attack & \multicolumn{2}{c}{MSE=0$\&$ SSIM=1} & MSE=0 \\
Multivariate attack &  \multicolumn{2}{c}{MSE=0$\&$ SSIM=1} & MSE=0 \\
Statistical attack &  \multicolumn{2}{c}{MSE=0$\&$ SSIM=1}& MSE=0 \\
  \bottomrule
\end{tabular}
\vspace{2mm}
\caption{Perfect input reconstruction results:  MSE and SSIM of 500 reconstructed samples and their corresponding training samples randomly selected in each dataset, using three mHSSP attacks.}
\label{tab.mse}
\end{table}

\begin{figure}[t!]
  \begin{subfigure}[t]{1.0\linewidth}
    \centering\includegraphics[width=80mm]{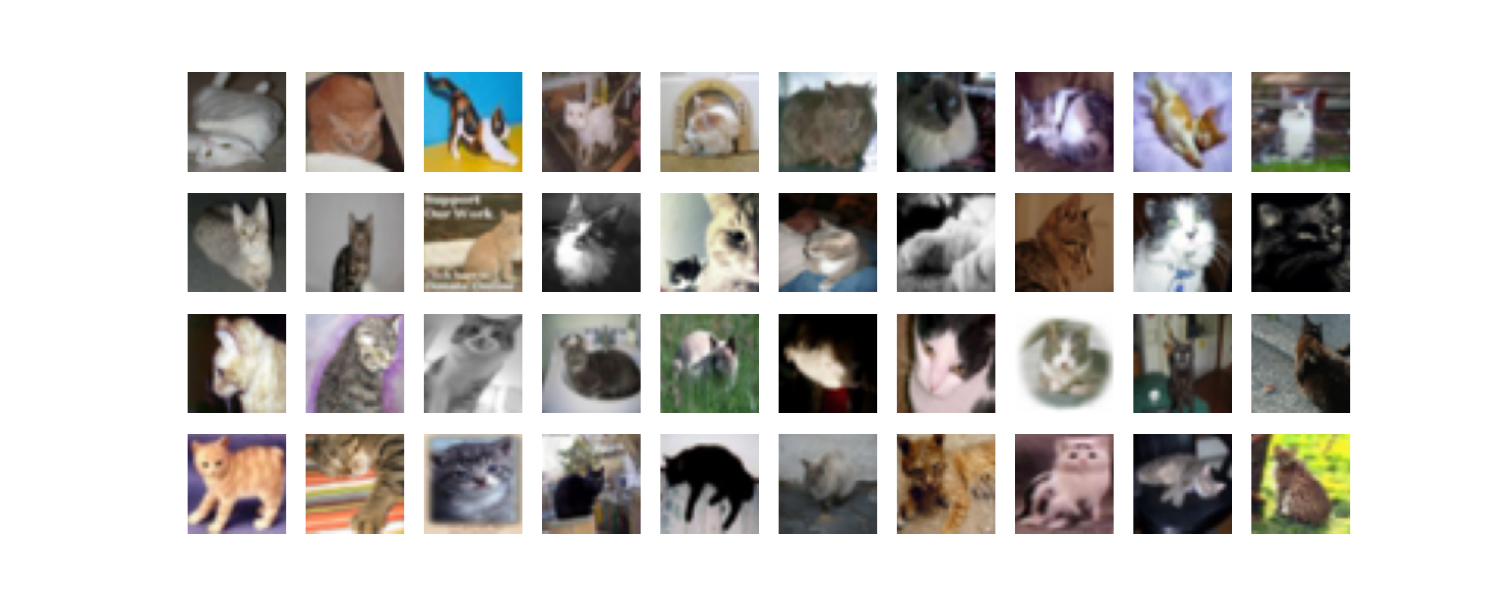}
    \centering\includegraphics[width=80mm]{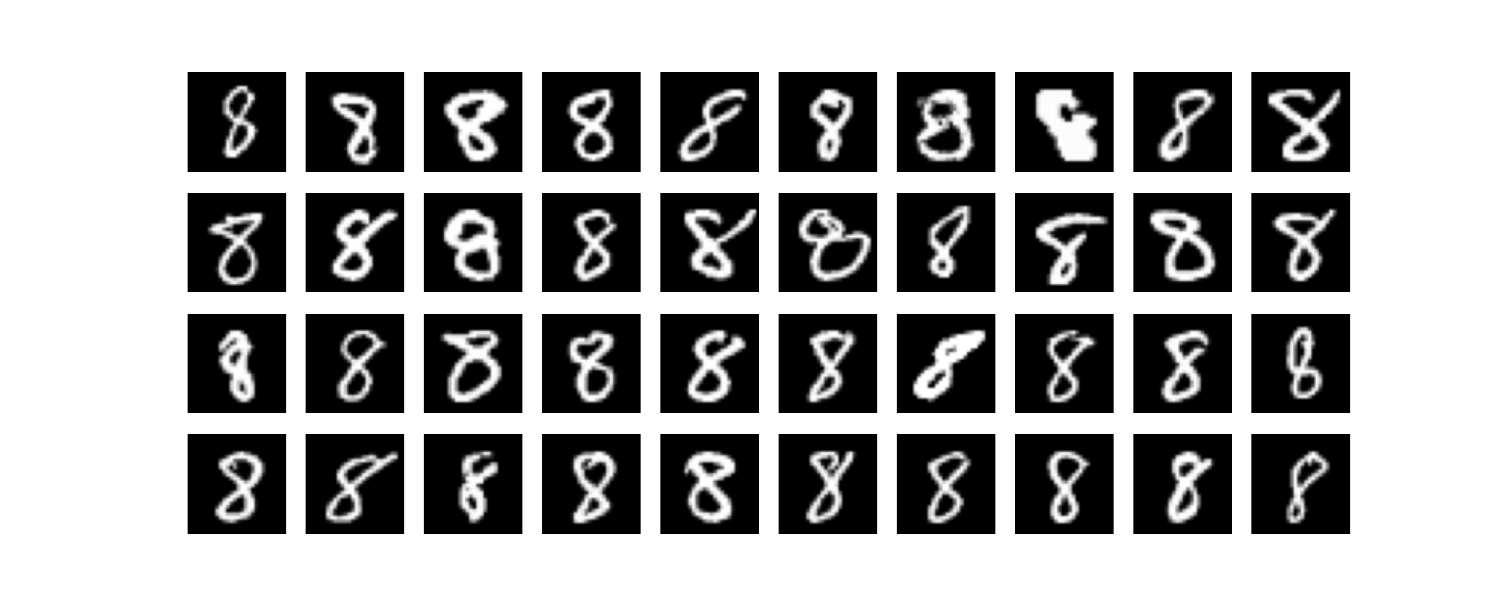}
    \caption{Ground truth: $B=40$ training samples randomly sampled from the same label: 'cat' in CIFAR-10 (top panel) and digit '8' in MNIST (bottom panel).}
  \end{subfigure}
    \begin{subfigure}[t]{1.0\linewidth}
    \centering\includegraphics[width=80mm]{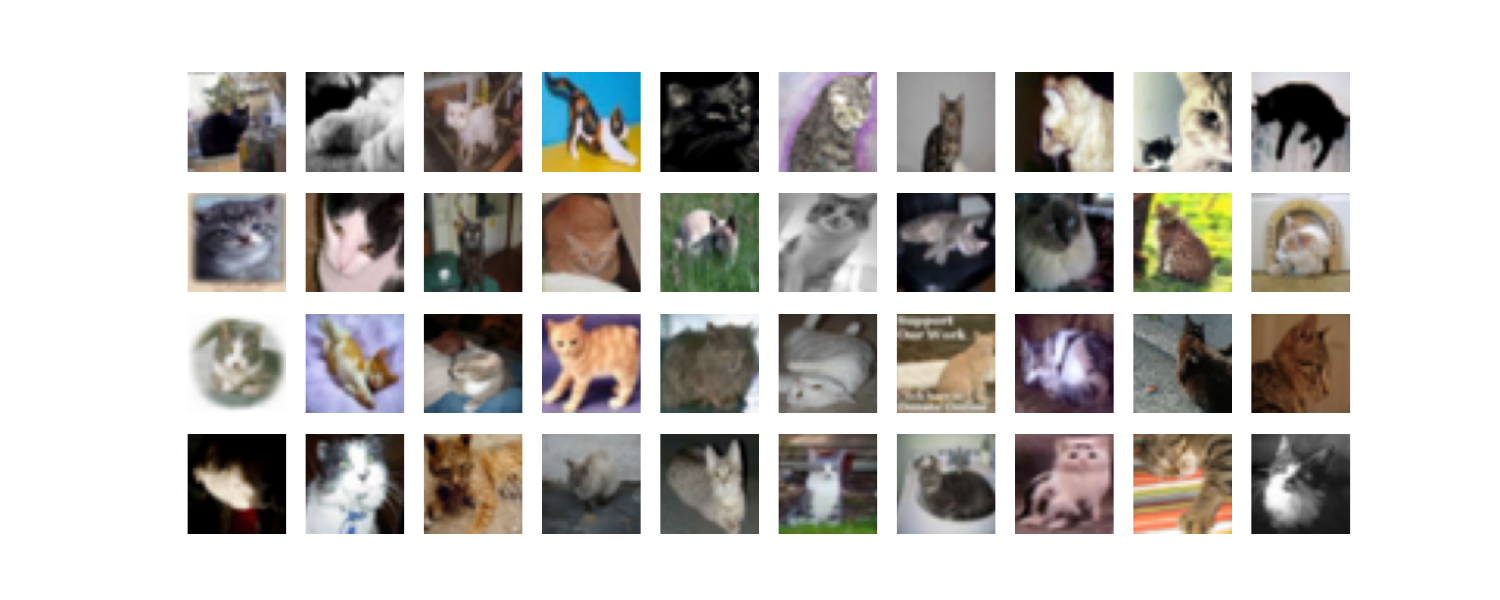}
    \centering\includegraphics[width=80mm]{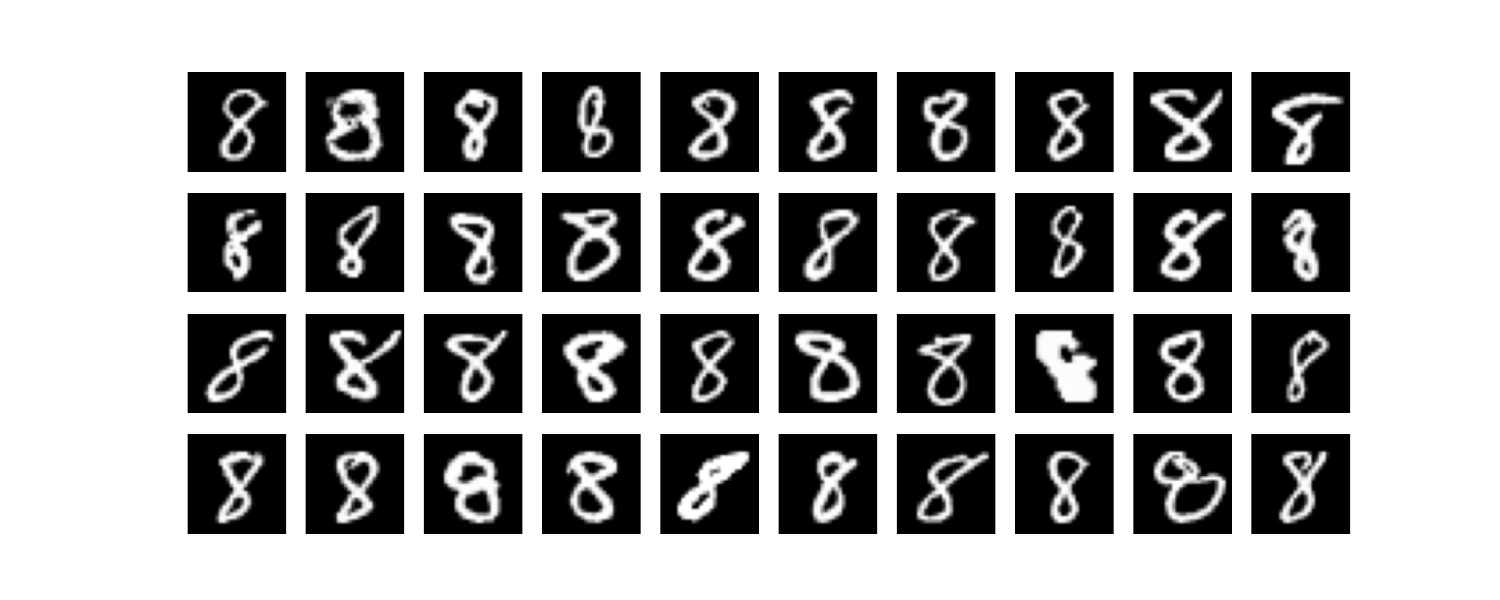}
    \caption{Reconstructed samples using the NS attack for each dataset.}
  \end{subfigure}
     \vskip -4pt
\caption{Input reconstruction results using mHSSP attacks for a batch of samples from the same class label. Due to space limits, results from the multivariate and statistical attacks are presented in Fig. \ref{fig.fl2c} in Appendix \ref{app.con}.}
\label{fig.fl2}
\vskip -4pt
\end{figure}


\begin{figure}
\begin{subfigure}[t]{1.0\linewidth}
    \centering\includegraphics[width=80mm]{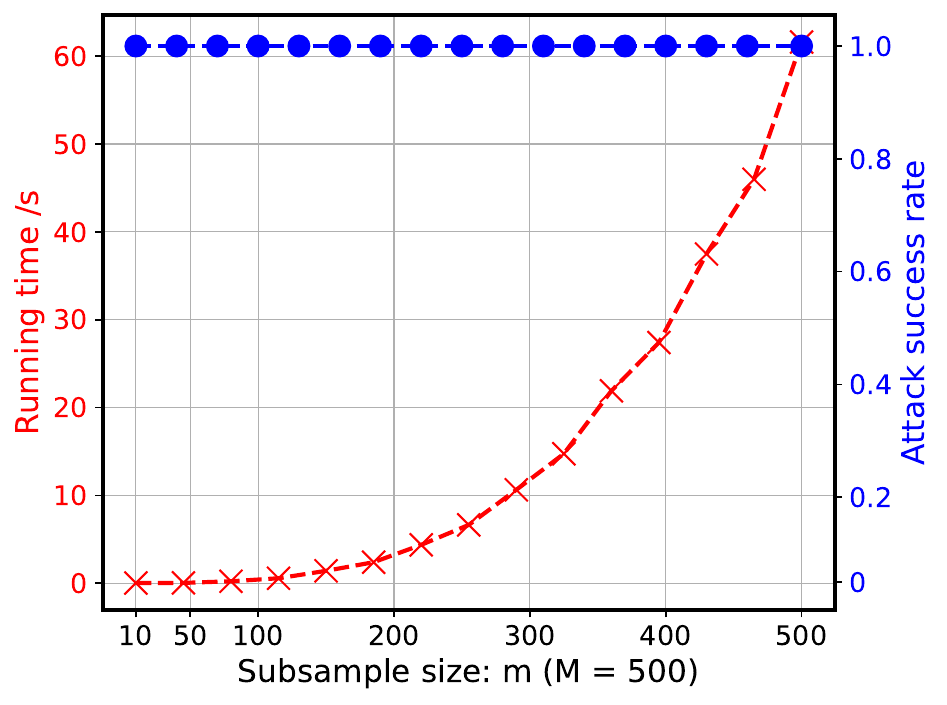}
    \vskip -4pt
    \caption{NS attack}
  \end{subfigure}
    \begin{subfigure}[t]{1.0\linewidth}
    \centering\includegraphics[width=80mm]{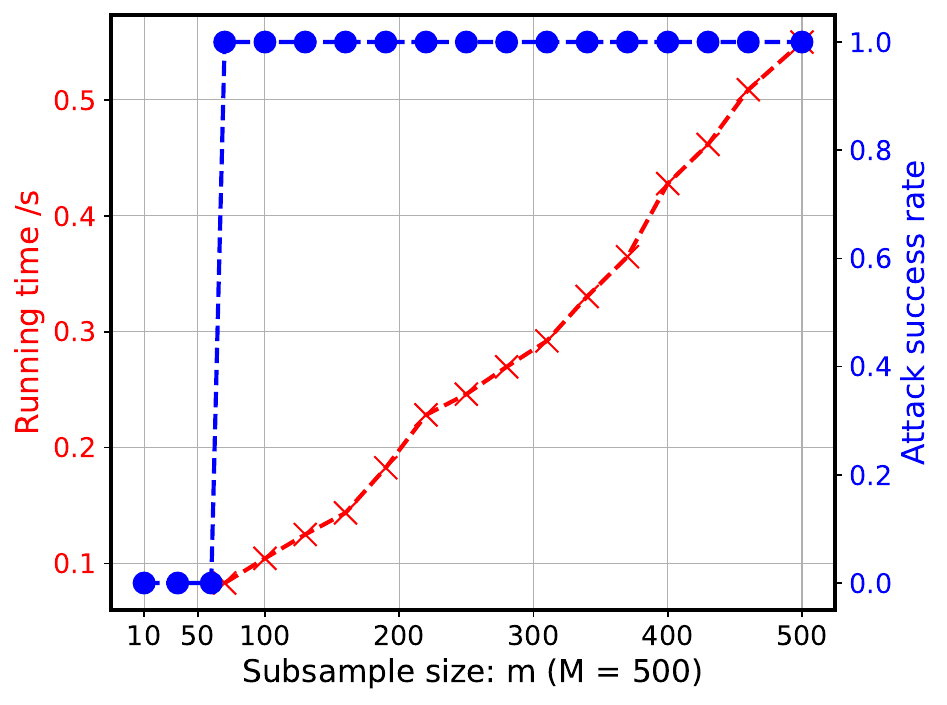}
    \vskip -4pt
    \caption{Multivariate attack}
  \end{subfigure}
\begin{subfigure}[t]{1.0\linewidth}
    \centering\includegraphics[width=80mm]{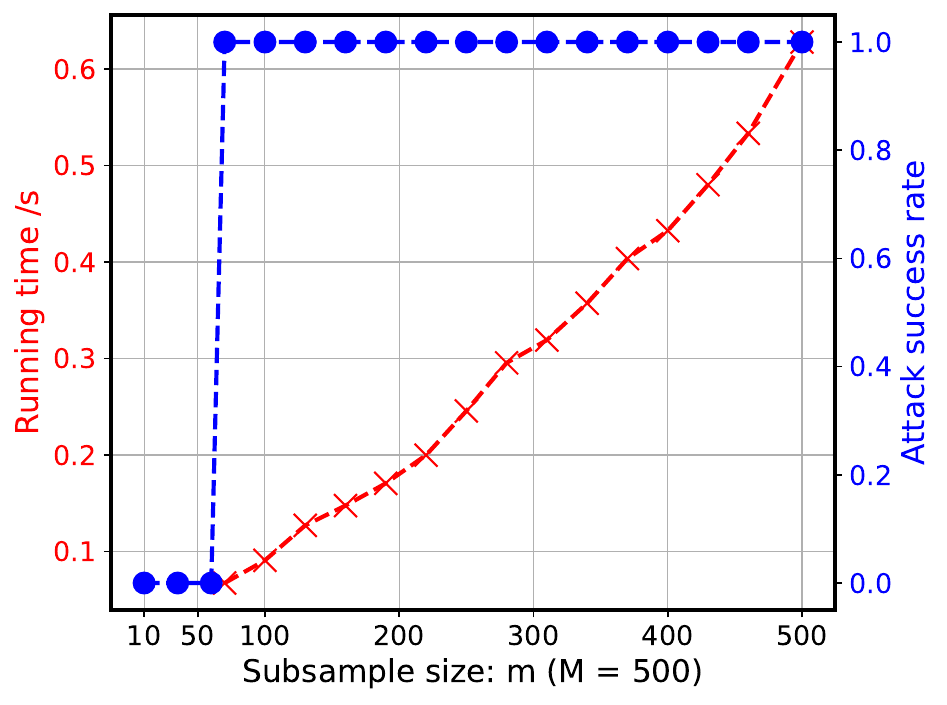}
    \vskip -4pt
    \caption{Statistical attack}
  \end{subfigure}
     \vskip -4pt
    \caption{ Running time (red lines) and attack success rate (blue lines) as a function of subsample size $m\leq M$ using three mHSSP attacks for CIFAR-10 dataset. }
    \label{fig.mHSSP}
    \vskip -4pt
\end{figure}

\begin{figure}
    \centering
    \includegraphics[width=.45\textwidth]{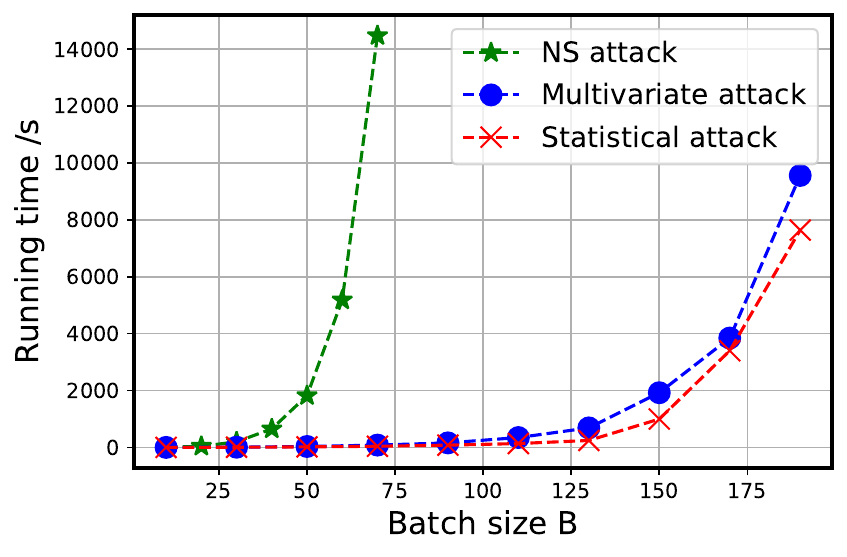}
    \caption{Running time as a function of the batch size $B$ using three mHSSP attacks. }
    \label{fig.timeHssp}
\end{figure}
\subsection{Validation of time complexity} \label{subsec.timeMB}
To substantiate the claims in Section \ref{subsec.time}, we now present how the time complexity changes regarding two important parameters: the batch size $B$ and the number of neurons $M$ in the first layer. \\ \hspace*{\fill} 

\noindent\textbf{Running time v.s., Number of neurons $M$:} 
As discussed in Section \ref{subsec.time}, the time complexity of mHSSP attacks increases significantly when \(M \gg B\). This challenge can be effectively managed by subsampling when 
 $M$ is excessively large, which can greatly improve the efficiency of the mHSSP attacks.  Define $m\leq M$ as the subsample size, i.e., we randomly subsample $m\leq M$ rows from the gradient matrix $\v G_w$ in \eqref{eq.Gw} and use it as input to the mHSSP attacks. A practical heuristic choice for $m$ is $m=2B$ for NS attack and $m\approx B^2$ for multivariate and statistical attacks \cite{coron2020polynomial,coron2021provably}. 

To examine this relationship, in Fig. \ref{fig.mHSSP} we demonstrate both the running time (red line) and attack success rate (blue line) as a function of subsample size $m\leq M$ given a fixed batch size $B$ using three mHSSP attacks for CIFAR dataset (see similar results for MNIST and Purchase datasets in Fig \ref{fig.mHSSP_mnist} and \ref{fig.mHSSP_purchase} of Appendix \ref{app.con}). 
For the sake of reconstruction efficiency, the batch size is set as 10, and for each dataset we trained 50 models to calculate the mean running time and attack success rate. For each model, a random batch was selected and subjected to mHSSP attacks using varying subsample sizes ranging from 10 to 500.  As expected, if $m=M$, the attacks are successful, but the running time is prohibitively high. In addition, using only a subset of the gradient matrix for attacks increases time efficiency without compromising the success rate. This efficiency is due to the requirement in Step 1 to compute the orthogonal lattice in dimension $m$, a larger $m$ demands more computational time. However, from the perspective of attack success rate, only a sufficient amount of information is necessary for effective reconstruction. Hence, we substantiate the conclusion that the number of neurons impacts the time complexity and the attack efficiency can be optimized by appropriate subsampling. 
  \\ \hspace*{\fill}

\noindent\textbf{Running time v.s., Batch size $B$:} 
After investigating the effect of neuron number $M$ in running time, we now examine the second parameter, namely the batch size $B$. 
In Figure \ref{fig.timeHssp} we demonstrate the running times associated with varying batch sizes $B$ for three mHSSP attacks using CIFAR-10. For each attack, we applied the above-described heuristic to optimize the subsample sizes, thereby improving attack efficiency.  The result confirms that, as anticipated, the more recent multivariate and statistical attacks are in general faster than the NS attack. Notably, the running time escalates rapidly as the batch size increases. This observation confirms the sensitivity of these attacks to changes in batch size, highlighting it as a critical parameter for reconstructing inputs from the gradients in FL. Hence, this also explains why existing empirical gradient inversion attacks have difficulty reconstructing inputs for big batch sizes.






\subsection{Limitations}
We remark that our work is only an initial effort to analyze the privacy implications of gradient sharing in FL using HSSP. Although our findings provide foundational insights, many interesting avenues for further research could enhance the breadth and depth of our understanding. A few of them are listed below:
\begin{enumerate}
\item \textbf{Expansion to other model architectures:} Currently, our analysis is limited to MLPs. It would be beneficial to extend this analysis to other architectures such as Convolutional Neural Networks (CNNs). Intuitively, reconstructing inputs from CNNs might be more challenging than from MLPs because CNNs tend to abstract and smooth input information, potentially complicating the inversion process.
\item  \textbf{Optimization of attacks using bias information:} The bias information, as represented in \eqref{eq.Gb}, while not directly derived from the input training data, contains insights about the weights of each row of the hidden weight matrix $\v D$. This aspect has not been explored in existing attacks, presenting an opportunity to develop more sophisticated attacks by leveraging this data.
\item \textbf{Exploration of time complexity related to data dimension parameter $u$:} In Section \ref{subsec.mhssp}, we discussed addressing the multidimensional nature of data when constructing $\mathcal{L}_Q^{\bot}(\mathbf{H})$. However, the dimension parameter $u$ has not been fully explored in the context of existing attacks, particularly concerning time complexity. This is a crucial factor for further investigation, especially given that $u$ is often large in machine learning applications.
\item  \textbf{Problem analysis in more complex settings:} As noted in Section \ref{subsec.conn}, our current analysis assumes a simplified scenario where the hidden weight matrix $\v D$ is binary, corresponding to the mHSSP. A more practical approach would consider non-binary weight matrices, thereby transitioning to a multidimensional Hidden Linear Combination Problem (mHLCP) where weights range from $[0,c]$. We hypothesize that the time complexity in such settings could significantly exceed that of mHSSP, particularly when the weight range $c$ is extensive. Further detailed investigation is required to confirm this hypothesis. 
\end{enumerate}

\section{Related work}\label{sec.prior}

\subsection{Privacy in FL}
Concerning privacy in machine learning models, there are mainly three primary types of privacy attacks:  1) Membership inference attacks \cite{shokri2017membership,salem2018ml,yeom2018privacy,song2021systematic} aim to determine whether a specific sample was used in the training dataset or not; 2) Property inference attack \cite{ateniese2015hacking,ganju2018property} seek to deduce properties of the input training data such as gender, class label; Input reconstruction attacks (e.g.,  model inversion attack \cite{fredrikson2014privacy,fredrikson2015model}) attempt to reconstruct the original training samples. All of these attacks have also been investigated in FL. It is shown that by exploiting the model updates shared to the server, an attacker can infer sensitive information such as 
the training data's properties \cite{xu2020subject,melis2019exploiting,wainakh2021user} and membership information \cite{melis2019exploiting}, or even allowing for the reconstruction of the input training samples \cite{zhu2019deep,zhao2020idlg,geiping2020inverting,yin2021see,boenisch2021curious,geng2023improved,fowl2021robbing,zhu2020r,wei2020framework,yang2022using,zhao2022deep,xu2022agic}. Our study concentrates on input reconstruction attacks due to their significant privacy implications and their prevalence in FL research. 

\subsection{Gradient inversion attack}
Gradient inversion attack is the most popular input reconstruction attack in FL as the shared information is often the updated gradients. \hspace*{\fill} 
 
\textbf{Optimization based approaches:} DLG \cite{zhu2019deep}  is the first proposed gradient inversion attack, which demonstrates that it is possible to reconstruct the training data with high fidelity using the shared gradient information. 
The main observation is that if two data points are similar to each other, their gradients might also be similar. DLG first initializes a random noise image called dummy data and produces a corresponding dummy gradient, it then iteratively optimizes the dummy data to resemble the training sample by minimizing the distance between the dummy gradient and the real gradient generated by the real training data. Despite its innovation, DLG struggled with convergence issues and was limited to low-resolution images and small batch sizes. 
To address these limitations, a rich line of work \cite{zhao2020idlg,geiping2020inverting,yin2021see,boenisch2021curious,geng2023improved,fowl2021robbing,zhu2020r,wei2020framework,yang2022using} try to propose more advanced gradient inversion attacks, e.g.,  increasing the optimization efficiency and producing accurate reconstructed data for high-resolution images \cite{zhao2020idlg,geiping2020inverting}. Yin et al. \cite{yin2021see} further improved the optimization efficiency by adding regularization and the attack can work with the batch size up to $48$. However, it requires that one batch cannot contain two images from the same class.  Yang et al. \cite{yang2022using} optimizes the attack efficiency and considers a more challenging case where the gradients are compressed.  

Given a reasonable batch size such as 128 the reconstructed data samples of existing work are often of low fidelity \cite{yin2021see,geng2023improved}.  Multiple reasons might cause such poor performances. 1) The optimization-based approaches suffer from the local minima problem and thus the reconstructed samples are often quite different from the original training data \cite{yin2021see}. 2) The fundamental assumption that minimizing the distance between gradients would make the synthesized data similar to the original training data is not necessarily true, as different mini-batch of data may produce almost identical gradients \cite{shumailov2021manipulating}.   In addition, a major challenge in input reconstruction attacks, such as gradient inversion and model inversion, involves accurately quantifying the similarity between two samples. For example, metrics like SSIM or PSNR are commonly used to evaluate the likeness between reconstructed and original inputs. However, these metrics are often considered limited and unreliable in accurately reflecting sample similarity, as highlighted in several studies \cite{almohammad2010stego,pambrun2015limitations}.  \\ \hspace*{\fill} 


\noindent\textbf{Analytical approaches} There are a few theoretical analyses on the privacy leakage caused by gradient sharing. Geiping et al. \cite{geiping2020inverting} theoretically prove that it is possible to perfectly reconstruct the training samples regardless of trained or untrained models. However, it only considers restrictive settings, e.g., only one sample is activated among a mini-batch.  Fowl et al \cite{fowl2021robbing} showed that it is possible to reconstruct the input by assuming a stronger threat model where the central server is active, i.e., it can change the model architecture.  Similarly, Boenisch et al. \cite{boenisch2021curious}  showed that perfect input reconstruction is possible by assuming that the central server can manipulate the weights maliciously. \\  \hspace*{\fill}  

\noindent\textbf{Distinctions from existing work} Our work differs from these existing approaches by adopting a cryptographic perspective, namely HSSP, to assess the difficulty of inverting input data from shared parameters (such as gradients and model weights) in FL. This approach not only overcomes empirical limitations, such as the need for specific class distribution in batches but also addresses issues related to the sensitivity of evaluation metrics, enabling perfect input reconstruction. Importantly, our analysis demonstrates that the difficulty of reconstructing inputs from gradients depends on the batch size, theoretically explaining the degradation in the performance of empirical attacks as batch sizes increase. Distinctively, our approach accommodates more general scenarios by not restricting to a single batch size, and rigorously derived the associated time complexities.  This novel cryptographic perspective offers a unique analytical angle that effectively addresses and mitigates the limitations observed in prior empirical and theoretical studies.

\section{Conclusion}\label{sec.conclu}
In this paper, we take the first step in theoretically analyzing input reconstruction attacks in Federated Learning (FL) from a cryptographic perspective through the lens of the Hidden Subset Sum Problem (HSSP). By mathematically formulating the gradients shared in an MLP network as a classical cryptographic problem called HSSP, we leveraged cryptographic tools to analyze it.   Our analyses demonstrated that the time complexity of reconstructing inputs in FL escalates with increasing batch sizes, a finding that elucidates the observed degradation in the performance of existing gradient inversion attacks as batch sizes expand.  Furthermore, we showed that adopting secure aggregation techniques would be a good defense against input reconstruction attacks as it directly increases the time complexity of existing HSSP attacks by $N^9$ times where $N$ is the number of participants in FL. We also discussed some interesting directions for further exploration. Our work not only sheds light on profound theoretical insights but also paves the way for future research at the intersection of cryptography and machine learning.   As an initial work employing cryptographic analysis to understand privacy vulnerabilities in FL, we call for more researchers to explore this promising cross-disciplinary topic.






%


\newpage
\bibliographystyle{IEEEtran}
\bibliography{main.bib}

\begin{thebibliography}{10}
\providecommand{\url}[1]{#1}
\csname url@samestyle\endcsname
\providecommand{\newblock}{\relax}
\providecommand{\bibinfo}[2]{#2}
\providecommand{\BIBentrySTDinterwordspacing}{\spaceskip=0pt\relax}
\providecommand{\BIBentryALTinterwordstretchfactor}{4}
\providecommand{\BIBentryALTinterwordspacing}{\spaceskip=\fontdimen2\font plus
\BIBentryALTinterwordstretchfactor\fontdimen3\font minus \fontdimen4\font\relax}
\providecommand{\BIBforeignlanguage}[2]{{%
\expandafter\ifx\csname l@#1\endcsname\relax
\typeout{** WARNING: IEEEtran.bst: No hyphenation pattern has been}%
\typeout{** loaded for the language `#1'. Using the pattern for}%
\typeout{** the default language instead.}%
\else
\language=\csname l@#1\endcsname
\fi
#2}}
\providecommand{\BIBdecl}{\relax}
\BIBdecl

\bibitem{anderson2015technology}
{M. Anderson}, \emph{Technology device ownership, 2015}.\hskip 1em plus 0.5em minus 0.4em\relax Pew Research Center, 2015.

\bibitem{poushter2016smartphone}
{J. Poushter and others}, ``Smartphone ownership and internet usage continues to climb in emerging economies,'' \emph{Pew Research Center}, vol.~22, pp. 1--44, 2016.

\bibitem{konevcny2016federateda}
J.~Kone{\v{c}}n{\'y}, H.~B. McMahan, F.~X. Yu, P.~Richtárik, A.~T. Suresh, and D.~Bacon, ``Federated learning: Strategies for improving communication efficiency,'' \emph{arXiv preprint arXiv:1610.05492}, 2016.

\bibitem{konevcny2016federatedb}
J.~Kone{\v{c}}n{\`y}, H.~B. McMahan, D.~Ramage, and P.~Richt{\'a}rik, ``Federated optimization: Distributed machine learning for on-device intelligence,'' \emph{arXiv preprint arXiv:1610.02527}, 2016.

\bibitem{mcmahan2016federated}
H.~B. McMahan, E.~Moore, D.~Ramage, and B.~A. y~Arcas, ``Federated learning of deep networks using model averaging,'' \emph{arXiv preprint arXiv:1602.05629}, vol.~2, 2016.

\bibitem{zhu2019deep}
L.~Zhu, Z.~Liu, and S.~Han, ``Deep leakage from gradients,'' \emph{Advances in neural information processing systems}, vol.~32, 2019.

\bibitem{zhao2020idlg}
B.~Zhao, K.~Mopuri, and H.~Bilen, ``i{DLG}: Improved deep leakage from gradients,'' \emph{arXiv preprint arXiv:2001.02610}, 2020.

\bibitem{geiping2020inverting}
J.~G. H.~Bauermeister, H.~Dr{\"o}ge and M.~Moeller, ``Inverting gradients-how easy is it to break privacy in federated learning?'' \emph{NeurIPS}, vol.~33, pp. 16\,937--16\,947, 2020.

\bibitem{yin2021see}
H.~Yin, A.~Mallya, A.~Vahdat, J.~Alvarez, J.~Kautz, and P.~Molchanov, ``See through gradients: Image batch recovery via gradinversion,'' in \emph{Proceedings of the IEEE/CVF Conference on Computer Vision and Pattern Recognition}, 2021, pp. 16\,337--16\,346.

\bibitem{boenisch2021curious}
F.~Boenisch, A.~Dziedzic, R.~Schuster, A.~S. Shamsabadi, I.~Shumailov, and N.~Papernot, ``When the curious abandon honesty: Federated learning is not private,'' in \emph{2023 IEEE 8th European Symposium on Security and Privacy (EuroS\&P)}.\hskip 1em plus 0.5em minus 0.4em\relax IEEE, 2023, pp. 175--199.

\bibitem{geng2023improved}
J.~Geng, Y.~Mou, Q.~Li, F.~Li, O.~Beyan, S.~Decker, and C.~Rong, ``Improved gradient inversion attacks and defenses in federated learning,'' \emph{IEEE Transactions on Big Data}, 2023.

\bibitem{fowl2021robbing}
L.~Fowl, J.~Geiping, W.~Czaja, M.~Goldblum, and T.~Goldstein, ``Robbing the fed: Directly obtaining private data in federated learning with modified models,'' \emph{arXiv preprint arXiv:2110.13057}, 2021.

\bibitem{zhu2020r}
J.~Zhu and M.~Blaschko, ``R-gap: Recursive gradient attack on privacy,'' \emph{arXiv preprint arXiv:2010.07733}, 2020.

\bibitem{wei2020framework}
W.~Wei, L.~Liu, M.~Loper, K.-H. Chow, M.~Gursoy, S.~Truex, and Y.~Wu, ``A framework for evaluating gradient leakage attacks in federated learning,'' \emph{arXiv preprint arXiv:2004.10397}, 2020.

\bibitem{yang2022using}
H.~Yang, M.~Ge, K.~Xiang, and J.~Li, ``Using highly compressed gradients in federated learning for data reconstruction attacks,'' \emph{IEEE Transactions on Information Forensics and Security}, vol.~18, pp. 818--830, 2022.

\bibitem{zhao2022deep}
Z.~Zhao, M.~Luo, and W.~Ding, ``Deep leakage from model in federated learning,'' \emph{arXiv preprint arXiv:2206.04887}, 2022.

\bibitem{xu2022agic}
J.~Xu, C.~Hong, J.~Huang, L.~Y. Chen, and J.~Decouchant, ``Agic: Approximate gradient inversion attack on federated learning,'' in \emph{2022 41st International Symposium on Reliable Distributed Systems (SRDS)}.\hskip 1em plus 0.5em minus 0.4em\relax IEEE, 2022, pp. 12--22.

\bibitem{wainakh2022federated}
A.~Wainakh, E.~Zimmer, S.~Subedi, J.~Keim, T.~Grube, S.~Karuppayah, G.~S. Alejandro, and M.~M{\"u}hlh{\"a}user, ``Federated learning attacks revisited: A critical discussion of gaps, assumptions, and evaluation setups,'' \emph{Sensors}, vol.~23, no.~1, p.~31, 2022.

\bibitem{agarap2018deep}
A.~Agarap, ``Deep learning using rectified linear units (relu),'' \emph{arXiv preprint arXiv:1803.08375}, 2018.

\bibitem{coron2020polynomial}
J.~Coron and A.~Gini, ``A polynomial-time algorithm for solving the hidden subset sum problem,'' in \emph{Advances in Cryptology--CRYPTO 2020}.\hskip 1em plus 0.5em minus 0.4em\relax Springer, 2020, pp. 3--31.

\bibitem{coron2021provably}
------, ``Provably solving the hidden subset sum problem via statistical learning,'' \emph{Mathematical Cryptology}, vol.~1, no.~2, pp. 70--84, 2021.

\bibitem{gini2022hardness}
A.~Gini, ``On the hardness of the hidden subset sum problem: algebraic and statistical attacks,'' Ph.D. dissertation, University of Luxembourg, Luxembourg, 2022.

\bibitem{dasgupta2008algorithms}
S.~Dasgupta, C.~H. Papadimitriou, and U.~V. Vazirani, \emph{Algorithms}.\hskip 1em plus 0.5em minus 0.4em\relax McGraw-Hill Higher Education New York, 2008.

\bibitem{almohammad2010stego}
A.~Almohammad and G.~Ghinea, ``Stego image quality and the reliability of psnr,'' in \emph{2010 2nd International Conference on Image Processing Theory, Tools and Applications}.\hskip 1em plus 0.5em minus 0.4em\relax IEEE, 2010, pp. 215--220.

\bibitem{pambrun2015limitations}
J.-F. Pambrun and R.~Noumeir, ``Limitations of the ssim quality metric in the context of diagnostic imaging,'' in \emph{2015 IEEE international conference on image processing (ICIP)}.\hskip 1em plus 0.5em minus 0.4em\relax IEEE, 2015, pp. 2960--2963.

\bibitem{Cramer2015}
{R. Cramer, I. B. Damg{\aa}rd, and J. B. Nielsen}, \emph{Secure Multiparty Computation and Secret Sharing}.\hskip 1em plus 0.5em minus 0.4em\relax Cambridge University Press, 2015.

\bibitem{paillier1999public}
{P. Paillier}, ``Public-key cryptosystems based on composite degree residuosity classes,'' in \emph{EUROCRYPT, pp. 223--238}, 1999.

\bibitem{mcmahan2017communication}
B.~McMahan, E.~Moore, D.~Ramage, S.~Hampson, and B.~A. y~Arcas, ``Communication-efficient learning of deep networks from decentralized data,'' in \emph{Proc. Int. Conf. Artif. Intell. Statist.}\hskip 1em plus 0.5em minus 0.4em\relax PMLR, 2017, pp. 1273--1282.

\bibitem{li2020federated}
{T. Li, A. K. Sahu, A. Talwalkar and V. Smith}, ``Federated learning: Challenges, methods, and future directions,'' \emph{IEEE Signal Process. Magazine, vol. 37, no. 3, pp. 50-60,}, 2020.

\bibitem{li2021model}
Q.~Li, B.~He, and D.~Song, ``Model-contrastive federated learning,'' in \emph{Proceedings of the IEEE/CVF conference on computer vision and pattern recognition}, 2021, pp. 10\,713--10\,722.

\bibitem{nguyen1999hardness}
P.~Nguyen and J.~Stern, ``The hardness of the hidden subset sum problem and its cryptographic implications,'' in \emph{Crypto}, vol.~99.\hskip 1em plus 0.5em minus 0.4em\relax Springer, 1999, pp. 31--46.

\bibitem{cassels1997introduction}
J.~Cassels, \emph{An introduction to the geometry of numbers (Reprint)}, ser. Classics in mathematics.\hskip 1em plus 0.5em minus 0.4em\relax Springer-Verlag, Berlin, 1997.

\bibitem{micciancio2002complexity}
D.~Micciancio and S.~Goldwasser, \emph{Complexity of lattice problems: a cryptographic perspective}.\hskip 1em plus 0.5em minus 0.4em\relax Springer Science \& Business Media, 2002, vol. 671.

\bibitem{lenstra1982factoring}
A.~Lenstra, H.~Lenstra, and L.~Lov{\'a}sz, ``Factoring polynomials with rational coefficients,'' \emph{Mathematische annalen}, vol. 261, no. ARTICLE, pp. 515--534, 1982.

\bibitem{chen2011bkz}
Y.~Chen and P.~Nguyen, ``Bkz 2.0: Better lattice security estimates,'' in \emph{Advances in Cryptology--ASIACRYPT 2011: 17th International Conference on the Theory and Application of Cryptology and Information Security, Seoul, South Korea, December 4-8, 2011. Proceedings 17}.\hskip 1em plus 0.5em minus 0.4em\relax Springer, 2011, pp. 1--20.

\bibitem{nguyen2006merkle}
P.~Nguyen and J.~Stern, ``Merkle-hellman revisited: a cryptanalysis of the qu-vanstone cryptosystem based on group factorizations,'' in \emph{Advances in Cryptology—CRYPTO'97: 17th Annual International Cryptology Conference Santa Barbara, California, USA August 17--21, 1997 Proceedings}.\hskip 1em plus 0.5em minus 0.4em\relax Springer, 2006, pp. 198--212.

\bibitem{chen2018computing}
J.~Chen, D.~Stehl{\'e}, and G.~Villard, ``Computing an lll-reduced basis of the orthogonal latice,'' in \emph{Proceedings of the 2018 ACM International Symposium on Symbolic and Algebraic Computation}, 2018, pp. 127--133.

\bibitem{nguyen2006learning}
P.~Q. Nguyen and O.~Regev, ``Learning a parallelepiped: Cryptanalysis of ggh and ntru signatures,'' in \emph{Advances in Cryptology-EUROCRYPT 2006: 24th Annual International Conference on the Theory and Applications of Cryptographic Techniques, St. Petersburg, Russia, May 28-June 1, 2006. Proceedings 25}.\hskip 1em plus 0.5em minus 0.4em\relax Springer, 2006, pp. 271--288.

\bibitem{cifar}
K.~Alex, H.~Geoffrey \emph{et~al.}, ``Learning multiple layers of features from tiny images,'' 2009.

\bibitem{xu2019verifynet}
G.~Xu, H.~Li, S.~Liu, K.~Yang, and X.~Lin, ``Verifynet: Secure and verifiable federated learning,'' \emph{IEEE Transactions on Information Forensics and Security}, vol.~15, pp. 911--926, 2019.

\bibitem{damgaard2012multiparty}
{I. Damg{\aa}rd, V. Pastro, N. Smart, and S. Zakarias}, ``Multiparty computation from somewhat homomorphic encryption,'' in \emph{Advances in Cryptology--CRYPTO, pp. 643--662}.\hskip 1em plus 0.5em minus 0.4em\relax Springer, 2012.

\bibitem{gentry2009fully}
{C. Gentry}, ``Fully homomorphic encryption using ideal lattices,'' in \emph{STOC, pp.169--178}, 2009.

\bibitem{deng2012mnist}
L.~Deng, ``The mnist database of handwritten digit images for machine learning research [best of the web],'' \emph{IEEE Signal Process. Mag.}, vol.~29, no.~6, pp. 141--142, 2012.

\bibitem{shokri2017membership}
{R. Shokri, M. Stronati, C. Song, and V. Shmatikov}, ``Membership inference attacks against machine learning models,'' in \emph{in Proc. IEEE Symp. Secur. Privacy (SP), pp.3--18}, 2017.

\bibitem{salem2018ml}
A.~Salem, Y.~Zhang, M.~Humbert, P.~Berrang, M.~Fritz, and M.~Backes, ``Ml-leaks: Model and data independent membership inference attacks and defenses on machine learning models,'' \emph{Network and Distributed Systems Security Symposium}, 2019.

\bibitem{yeom2018privacy}
S.~Yeom, I.~Giacomelli, M.~Fredrikson, and S.~Jha, ``Privacy risk in machine learning: Analyzing the connection to overfitting,'' in \emph{IEEE computer security foundations symposium (CSF)}, 2018, pp. 268--282.

\bibitem{song2021systematic}
L.~Song and P.~Mittal, ``Systematic evaluation of privacy risks of machine learning models,'' in \emph{USENIX Security Symposium (USENIX)}, 2021, pp. 2615--2632.

\bibitem{ateniese2015hacking}
G.~Ateniese, L.~Mancini, A.~Spognardi, A.~Villani, D.~Vitali, and G.~Felici, ``Hacking smart machines with smarter ones: How to extract meaningful data from machine learning classifiers,'' \emph{International Journal of Security and Networks}, vol.~10, no.~3, pp. 137--150, 2015.

\bibitem{ganju2018property}
G.~Karan, Q.~Wang, W.~Yang, C.~Gunter, and N.~Borisov, ``Property inference attacks on fully connected neural networks using permutation invariant representations,'' in \emph{Proceedings of the 2018 ACM SIGSAC conference on computer and communications security}, 2018, pp. 619--633.

\bibitem{fredrikson2014privacy}
M.~Fredrikson, E.~Lantz, S.~Jha, S.~Lin, D.~Page, and T.~Ristenpart, ``Privacy in pharmacogenetics: An end-to-end case study of personalized warfarin dosing,'' in \emph{23rd $\{$USENIX$\}$ Security Symposium}, 2014, pp. 17--32.

\bibitem{fredrikson2015model}
M.~Fredrikson, S.~Jha, and T.~Ristenpart, ``Model inversion attacks that exploit confidence information and basic countermeasures,'' in \emph{Proceedings of the 22nd ACM SIGSAC conference on computer and communications security}, 2015, pp. 1322--1333.

\bibitem{xu2020subject}
X.~L. M.~Xu, ``Subject property inference attack in collaborative learning,'' in \emph{2020 12th International Conference on Intelligent Human-Machine Systems and Cybernetics (IHMSC)}, vol.~1.\hskip 1em plus 0.5em minus 0.4em\relax IEEE, 2020, pp. 227--231.

\bibitem{melis2019exploiting}
L.~Melis, C.~Song, E.~D. Cristofaro, and V.~Shmatikov, ``Exploiting unintended feature leakage in collaborative learning,'' in \emph{2019 IEEE symposium on security and privacy (SP)}.\hskip 1em plus 0.5em minus 0.4em\relax IEEE, 2019, pp. 691--706.

\bibitem{wainakh2021user}
A.~Wainakh, F.~Ventola, T.~M{\"u}{\ss}ig, J.~Keim, C.~G. Cordero, E.~Zimmer, T.~Grube, K.~Kersting, and M.~M{\"u}hlh{\"a}user, ``User label leakage from gradients in federated learning,'' \emph{arXiv preprint arXiv:2105.09369}, 2021.

\bibitem{shumailov2021manipulating}
I.~Shumailov, Z.~Shumaylov, D.~Kazhdan, Y.~Zhao, N.~Papernot, M.~Erdogdu, and R.~Anderson, ``Manipulating sgd with data ordering attacks,'' \emph{Advances in Neural Information Processing Systems}, vol.~34, pp. 18\,021--18\,032, 2021.

\end{thebibliography}

\appendix

\section{A general way to construct a basis of $\mathcal{L}_Q^{\bot}(\mathbf{h})$ } \label{app.lattice}
Besides the case of $\gcd(h_1, Q)=1$  discussed before, in what follows we discuss how to construct the basis of $\mathcal{L}_Q^{\bot}(\mathbf{h})$ for other cases:

\noindent \textbf{The case of  $\mathbf{h} \equiv \mathbf{0} \mod Q$}: in this case the basis can be directly constructed as $\mathcal{L}_Q^{\bot}(\mathbf{h})=\mathbb{Z}^M$.

\noindent \textbf{The case of $\gcd(h_1, Q)=\gcd(h_1,h_2,\cdots,h_M,Q)=d$} where $d\neq 1$: the construction is the same as the case of $\gcd(h_1, Q)=1$ by simply replacing $Q$ and $\v h$ with $Q/d$ and $\v h /d$, respectively. Hence,  the basis of $\mathcal{L}_Q^{\bot}(\mathbf{h})$ can be constructed using the same basis of $\mathcal{L}_{Q/d}^{\bot}(\mathbf{h}/d)$.

\noindent \textbf{The case of $\gcd(h_1,h_2,\cdots,h_M,Q)=d$ while $\gcd(h_i,Q)=d_i\ne d,\ \forall i\in \{1,\dots,M\}$}:  define the first vector as  $\mathbf{y}_1:=(Q/d_1,0,\dots,0)$, the second one is constructed as  $\mathbf{y}_2:=(-(h_2/d_{12})\cdot ((h_1/d_1)^{-1} \mod Q/d_1),d_1/d_{12},0,\dots,0)$ where $d_{12}=\gcd(d_1,d_2)$. This construction ensures  $\langle \mathbf{y}_2,\mathbf{h}\rangle \equiv 0 \mod Q$ as 
\begin{align*}
\begin{split} \langle \mathbf{y}_2,\mathbf{h}\rangle &\equiv -(h_2/d_{12})\cdot ((h_1/d_1)^{-1}\mod Q/d_1)\cdot h_1\\
&\quad+d_1h_2/d_{12}\\
&\equiv -(h_2/d_{12})\cdot ((h_1/d_1)^{-1}\mod Q/d_1)\cdot(h_1/d_1)\cdot d_1\\
&\quad+d_1h_2/d_{12}\\
&\equiv -(h_2/d_{12})\cdot d_1+d_1h_2/d_{12}\\
&\equiv 0 \mod Q.
\end{split}
\end{align*}
The third vector can be constructed as $\mathbf{y}_3:=(-(h_3/d_{123})\cdot ((h_1/d_1)^{-1} \mod Q/d_1)\cdot k_1,-(h_3/d_{123})\cdot ((h_2/d_2)^{-1} \mod Q/d_2)\cdot k_2, d_{12}/d_{123},0,\dots,0)$ where $d_{123}=\gcd(d_1,d_2,d_3)$ and $k_1,k_2\in \mathbb{Z}$ satisfying $k_1d_1+k_2d_2=d_{12}$. It could be checked that $\langle \mathbf{y}_3,\mathbf{h}\rangle\equiv 0 \mod Q.$ Hence, the rest can be constructed similarly and collect them together as $[\v y_1,\v y_2,\ldots,\v y_M]$, which is a basis of $\mathcal{L}_Q^{\bot}(\mathbf{h})$.

\newpage
\section{Additional experimental results}\label{app.con}
\begin{figure}[H]
  \begin{subfigure}[t]{1.0\linewidth}
    \centering\includegraphics[width=80mm]{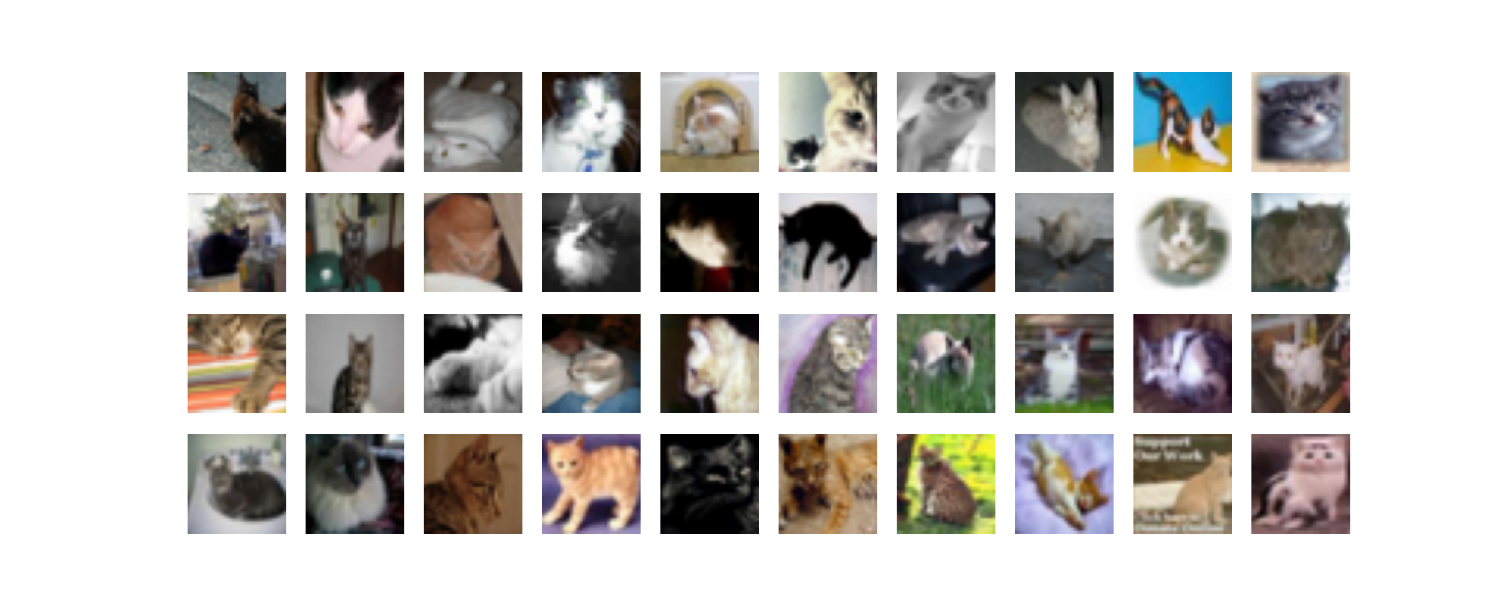}
    \centering\includegraphics[width=80mm]{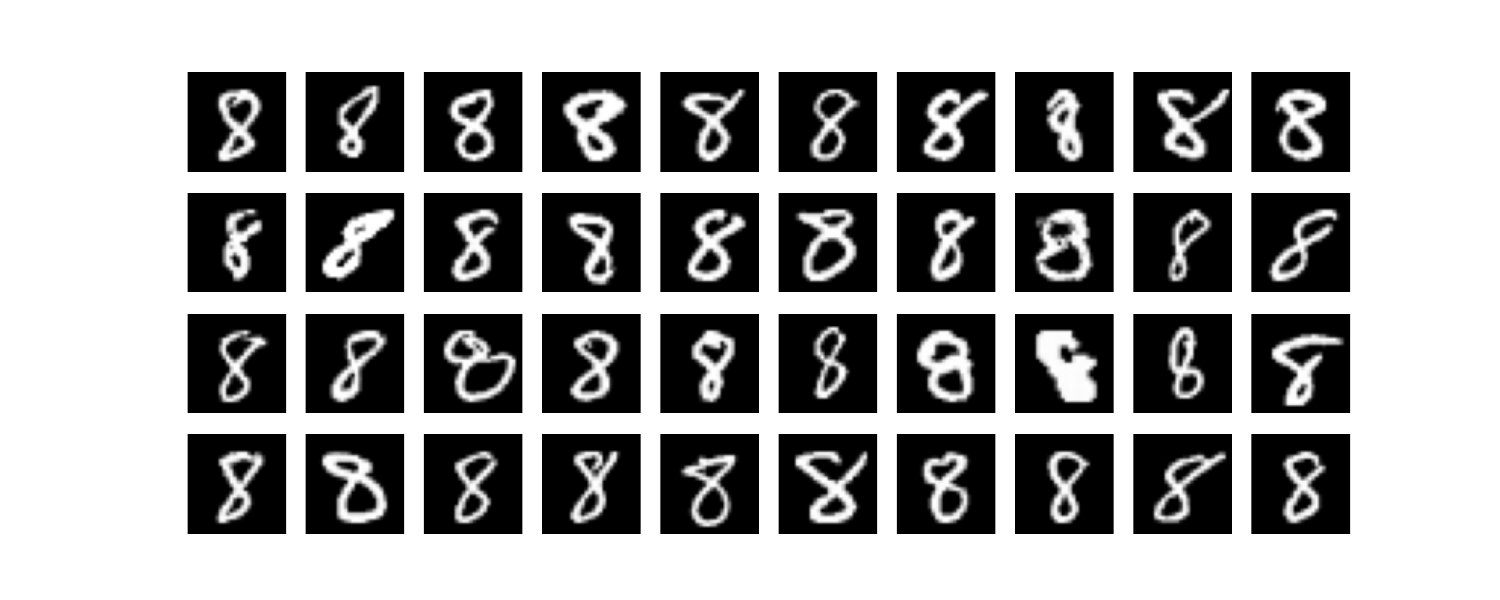}
    \caption{Reconstructed samples using the multivariate attack for each dataset.}
  \end{subfigure}
    \begin{subfigure}[t]{1.0\linewidth}
    \centering\includegraphics[width=80mm]{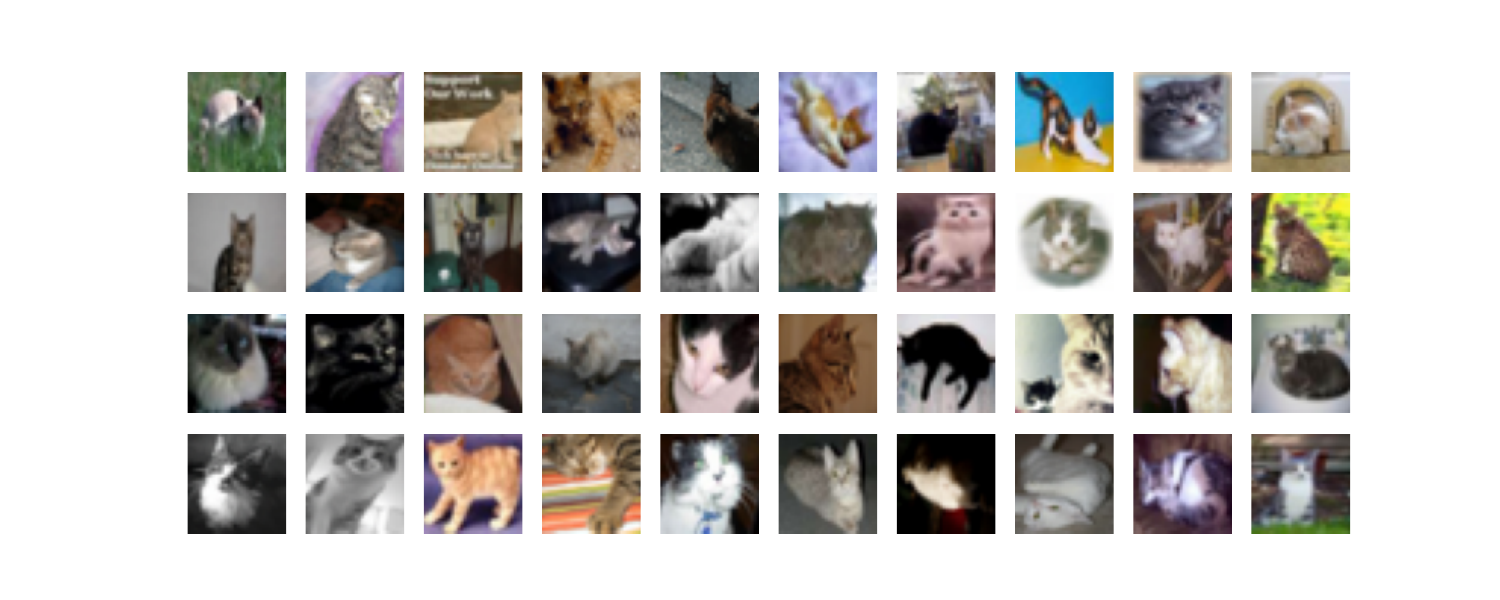}
    \centering\includegraphics[width=80mm]{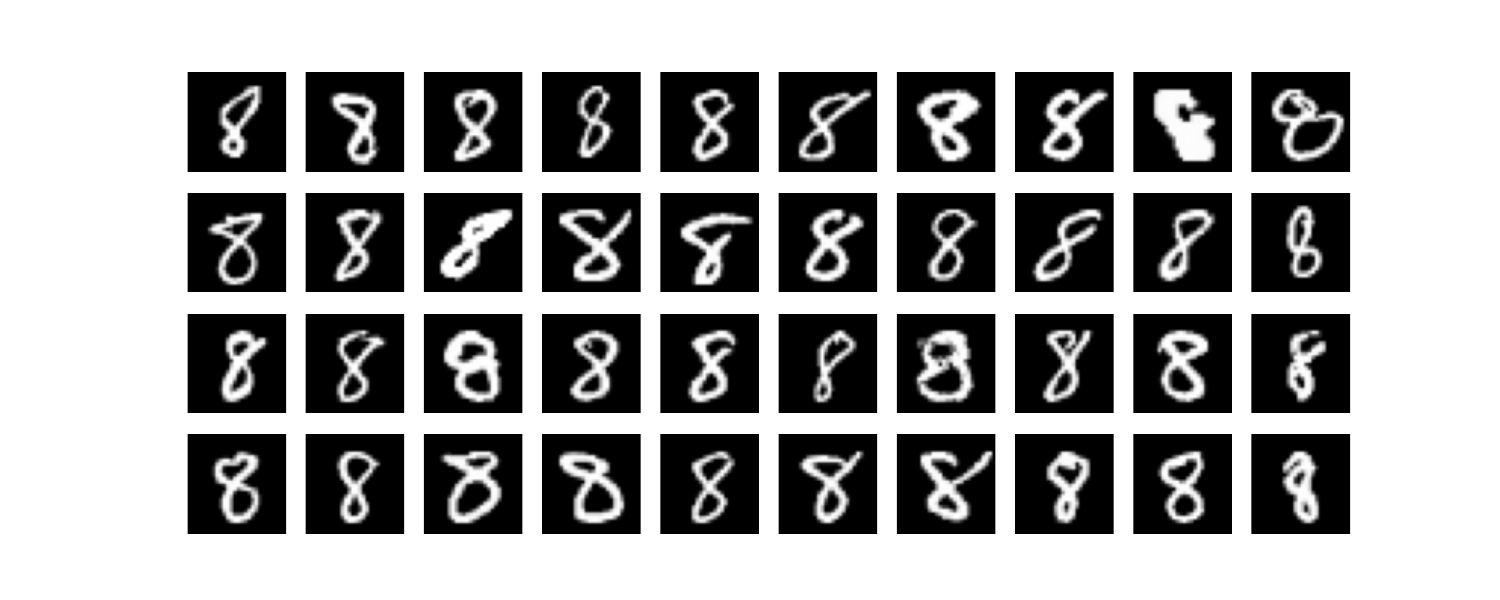}
    \caption{Reconstructed samples using the statistical attack for each dataset.}
  \end{subfigure}
\caption{Continuation of Figure \ref{fig.fl2}}
\label{fig.fl2c}
\end{figure}
\begin{figure*}[t!]
  \begin{subfigure}[t]{1.0\linewidth}
    \centering\includegraphics[width=130mm]{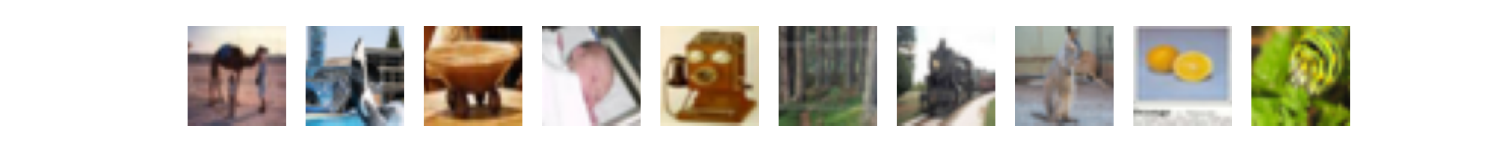}
    \centering\includegraphics[width=130mm]{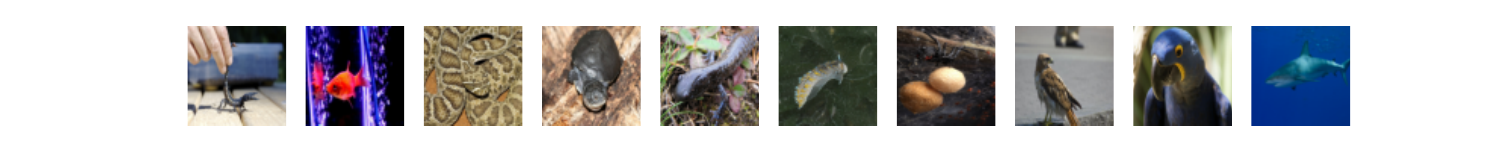}
    \captionsetup{justification=centering}
    \caption{Ground truth: $10$ training samples randomly selected from  CIFAR-100 (top panel) and Imagenet (Bottom panel).}
  \end{subfigure}
    \begin{subfigure}[t]{1.0\linewidth}
    \centering\includegraphics[width=130mm]{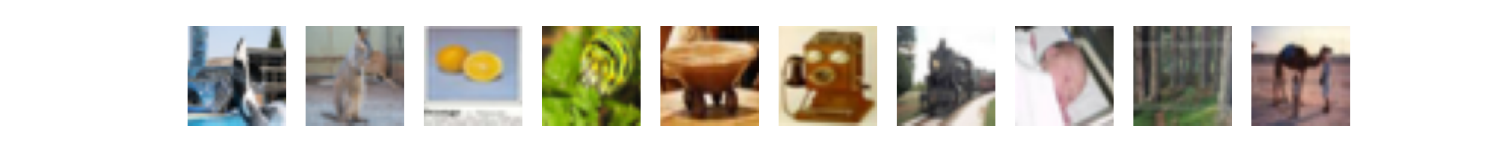}
    \centering\includegraphics[width=130mm]{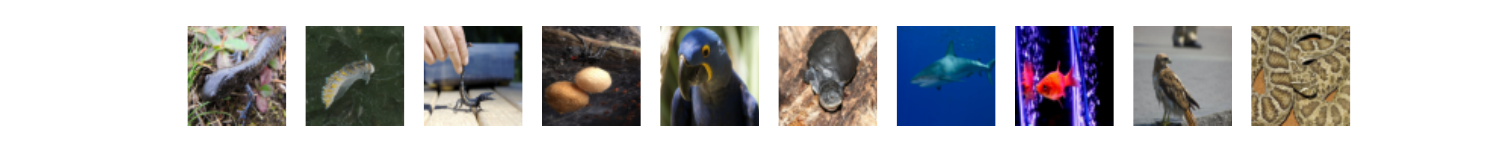}
    \caption{Reconstructed samples using the NS attack for each dataset.}
  \end{subfigure}
    \begin{subfigure}[t]{1.0\linewidth}
    \centering\includegraphics[width=130mm]{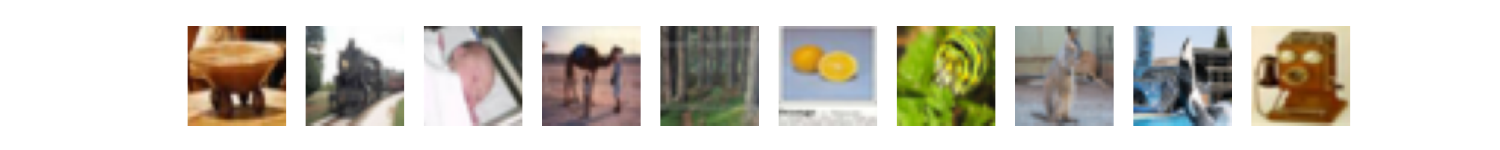}
    \centering\includegraphics[width=130mm]{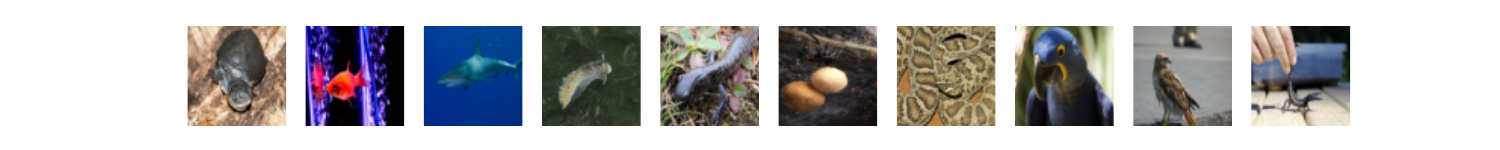}
    \caption{Reconstructed samples using the multivariate attack for each dataset.}
  \end{subfigure}
    \begin{subfigure}[t]{1.0\linewidth}
    \centering\includegraphics[width=130mm]{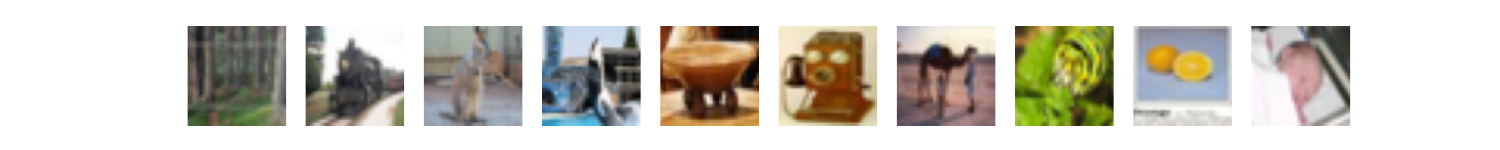}
    \centering\includegraphics[width=130mm]{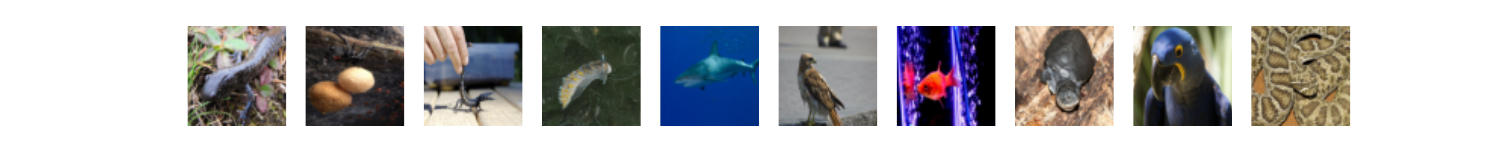}
    \caption{Reconstructed samples using the statistical attack for each dataset.}
  \end{subfigure}
\caption{Visualization examples of input reconstruction results using three mHSSP attacks for CIFAR-100 and Imagenet, respectively.}
\label{fig.imgnet}
\end{figure*}

\newpage

\begin{figure}
\begin{subfigure}[t]{1.0\linewidth}
    \centering\includegraphics[width=60mm]{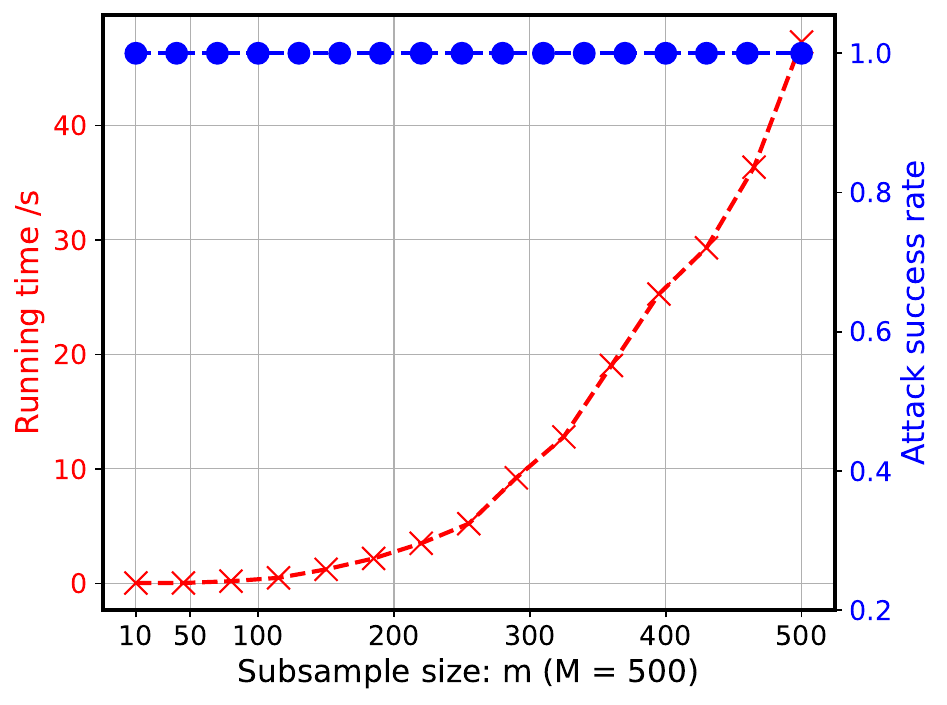}
    \vskip -4pt
    \caption{NS attack}
  \end{subfigure}
    \begin{subfigure}[t]{1.0\linewidth}
    \centering\includegraphics[width=60mm]{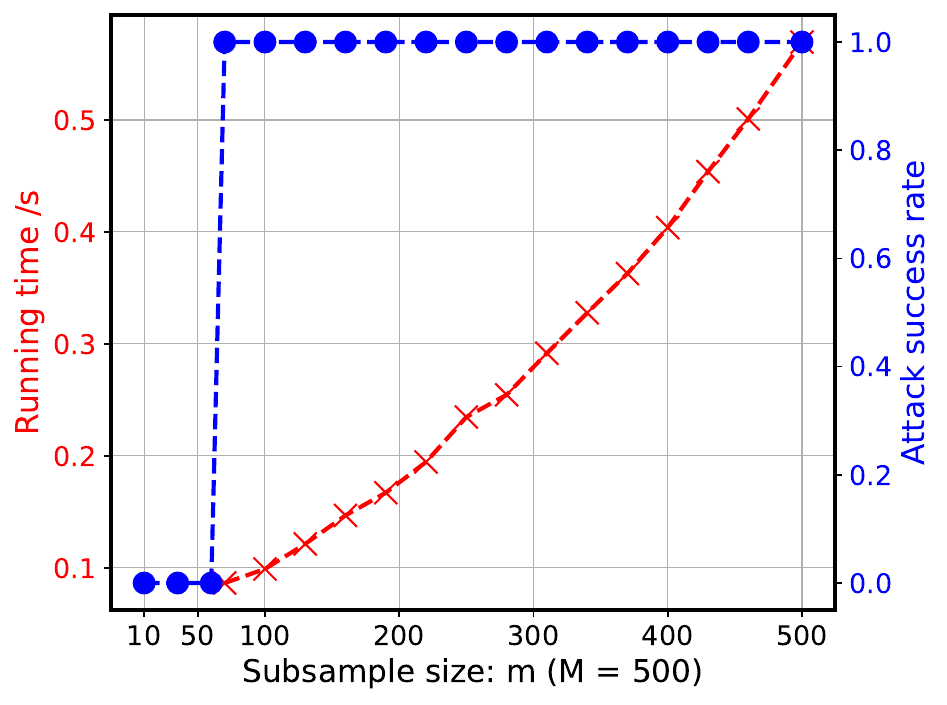}
    \vskip -4pt
    \caption{Multivariate attack}
  \end{subfigure}
\begin{subfigure}[t]{1.0\linewidth}
    \centering\includegraphics[width=60mm]{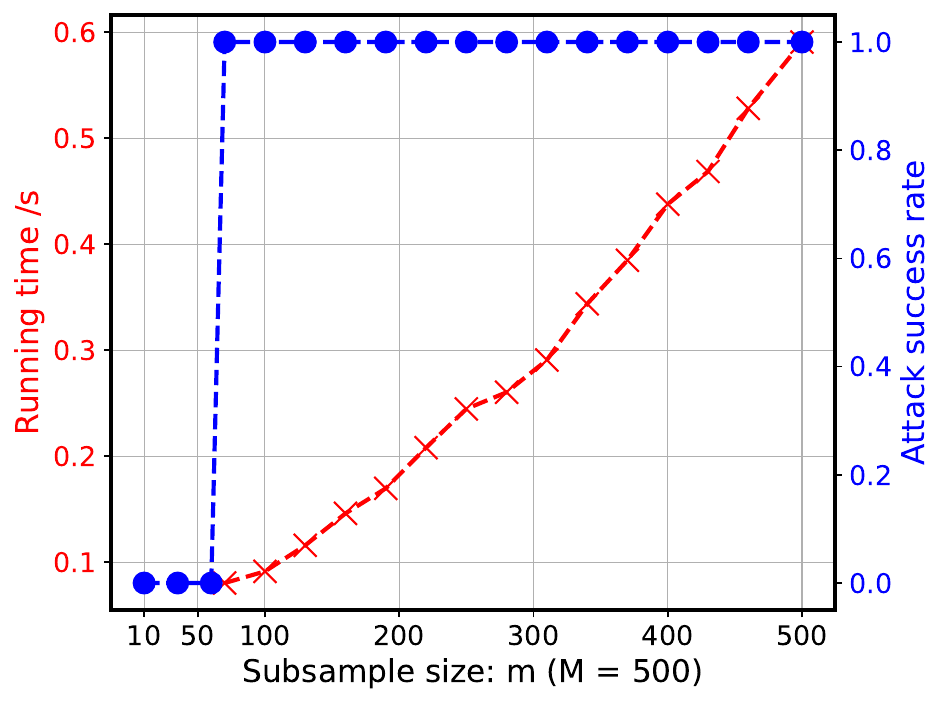}
    \vskip -4pt
    \caption{Statistical attack}
  \end{subfigure}
    \caption{ Running time (red lines) and attack success rate (blue lines) as a function of subsample size $m\leq M$ using three mHSSP attacks for MNIST dataset. }
    \label{fig.mHSSP_mnist}
\end{figure}

\begin{figure}
\begin{subfigure}[t]{1.0\linewidth}
    \centering\includegraphics[width=60mm]{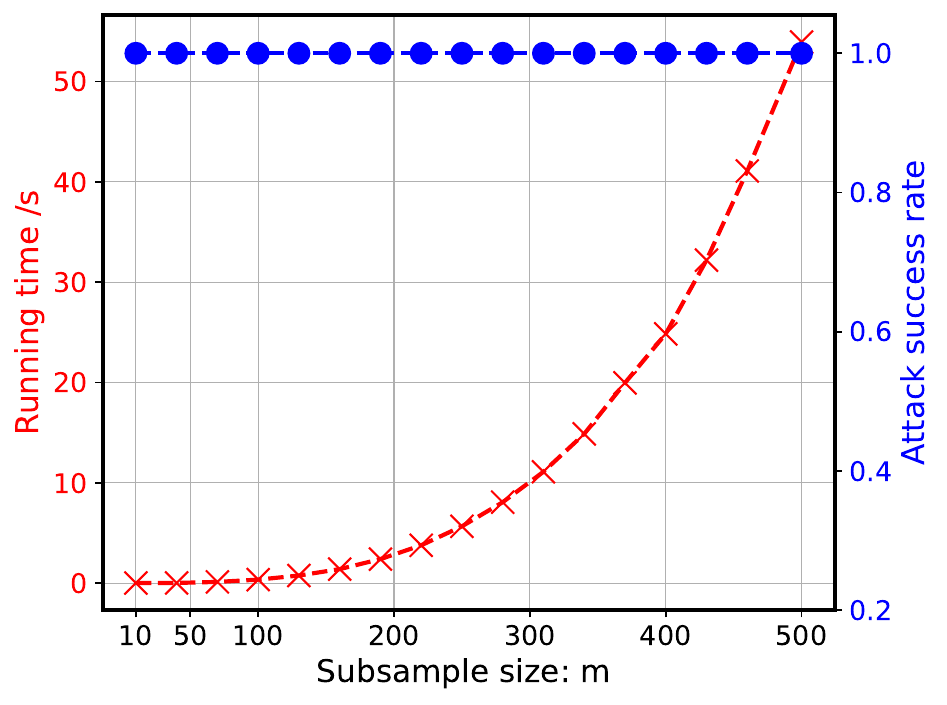}
    \vskip -4pt
    \caption{NS attack}
  \end{subfigure}
    \begin{subfigure}[t]{1.0\linewidth}
    \centering\includegraphics[width=60mm]{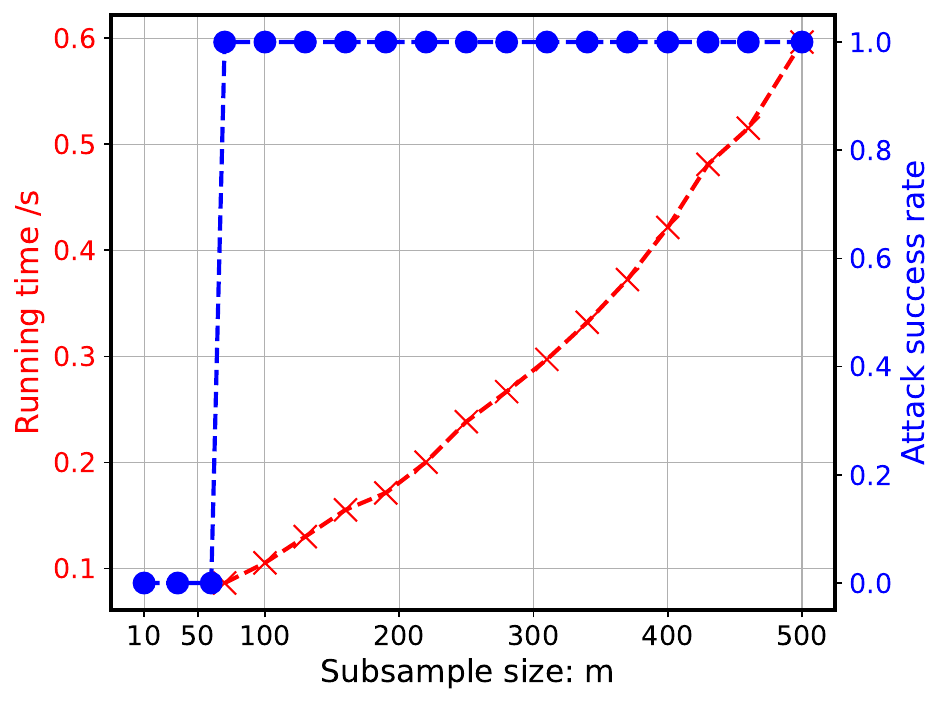}
    \vskip -4pt
    \caption{Multivariate attack}
  \end{subfigure}
\begin{subfigure}[t]{1.0\linewidth}
    \centering\includegraphics[width=60mm]{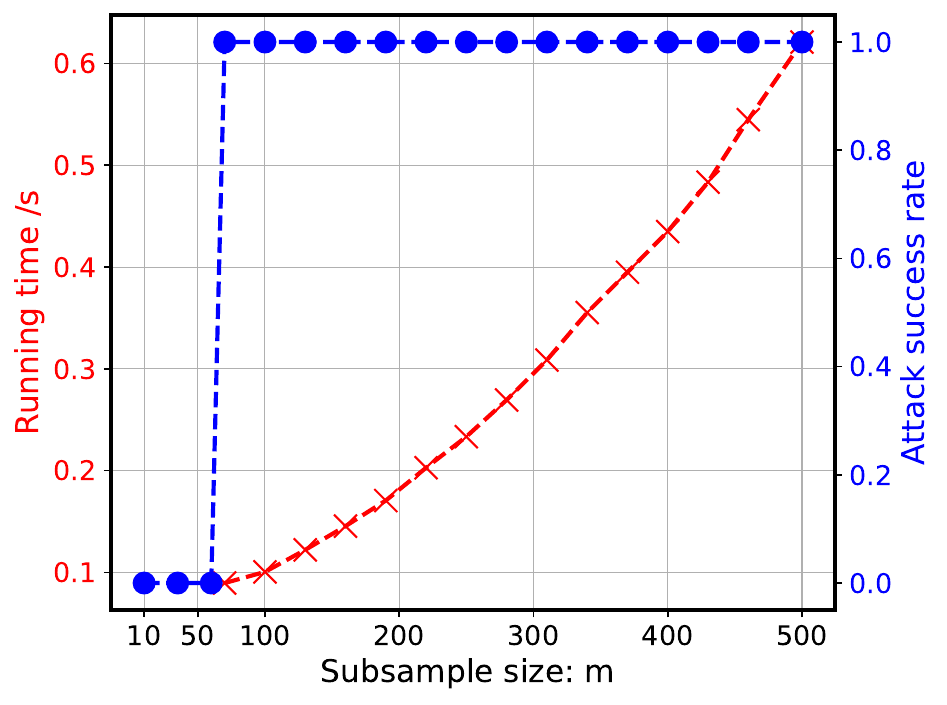}
    \vskip -4pt
    \caption{Statistical attack}
  \end{subfigure}
     \vskip -4pt
    \caption{ Running time (red lines) and attack success rate (blue lines) as a function of subsample size $m\leq M$ using three mHSSP attacks for Purchase dataset. }
    \label{fig.mHSSP_purchase}
    \vskip -4pt
\end{figure}

\end{document}